\documentclass[a4paper,aip,jmp,onecolumn,floatfix,reprint,11pt,tightenlines]{revtex4-2} 

\usepackage{amsmath}
\usepackage{amssymb}
\usepackage{bbm}
\usepackage{amsthm}
\usepackage{graphicx}

\usepackage{stmaryrd}

\usepackage[dvipsnames,x11names]{xcolor}

\usepackage{hyperref}
\hypersetup{
        pdftitle = {SIC-POVMs from Stark Units: Prime dimensions n^2+3},
        pdfauthor = {Marcus Appleby, Ingemar Bengtsson, Markus Grassl,
        Michael Harrison, Gary McConnell}
}

\usepackage{enumerate}

\usepackage{url}
\hypersetup{breaklinks}

\usepackage{mathtools}
\usepackage{mathrsfs}
\usepackage{amsmath,amscd,amsthm,multirow,bm}
\usepackage{multirow,bm}
\usepackage{array,booktabs}
\usepackage{cancel}
\usepackage{comment}

\usepackage{graphics}

\newenvironment{dedication}
        {\begin{quotation}\begin{center}\begin{em}}
        {\par\end{em}\end{center}\end{quotation}}

\DeclareMathOperator*{\lcm}{lcm}

\DeclareMathOperator{\coker}{coker}
\DeclareMathOperator{\im}{im}

\newtheorem{theorem}{Theorem}
\newtheorem{lemma}[theorem]{Lemma}
\newtheorem{corollary}[theorem]{Corollary}
\newtheorem{proposition}[theorem]{Proposition}
\newtheorem{remark}[theorem]{Remark}
\newtheorem{conjecture}{Conjecture}
\newtheorem{hypothesis}[]{Hypothesis}

\def\CC{\mathbb{C}}
\def\FF{\mathbb{F}}
\def\LL{\mathbb{L}}
\def\NN{\mathbb{N}}
\def\QQ{\mathbb{Q}}
\def\RR{\mathbb{R}}
\def\ZZ{\mathbb{Z}}
\def\K{\mathbb{K}}

\def\qq{\mathfrak{q}}
\def\QQQ{\mathfrak{Q}}

\def\clK{\mathcal{C}_K}
\let\iso\cong
\newcommand{\Syl}[2]{{\mathrm{Syl}_{#1}{#2}}}

\def\JJ{\mathcal{J}}

\def\calF{\mathcal{F}}

\def\qmuk{\sqrt{-u_K}}
\def\es{\epsilon_\sigma}
\def\est{\es}
\def\qea{{\sqrt{\est}}}

\def\dd{\partial}
\def\de{\partial_\ell}

\def\sfc{\mathfrak{p}_2}

\def\mm{\mathfrak{m}}

\def\jj{\mathfrak{j}}
\def\ff{\mathfrak{f}}
\def\pp{\mathfrak{p}}
\def\PP{\mathfrak{P}}

\def\sigt{{\sigma_{\text{\tiny{\!$T$}}}\,}}

\def\xea{\xi_\ell}

\newcommand{\ag}[1]{{\ZZ_K}/{#1}}
\newcommand{\mg}[1]{{\left(\ag{#1}\right)^\times}}

\newcommand{\leg}[2]{({\scriptstyle{\frac{#1}{#2}}})}
\newcommand{\ord}[2]{{\hbox{\rm ord}}_{#1} {#2} }
\newcommand{\nkq}[1]{{\norm{K}{\QQ}{#1}}}

\newcommand{\Gal}[2]{{\mathrm{Gal}_{#1/#2}}}
\newcommand{\norm}[3]{{\hbox{\small{\rm{N}}}_{\scriptscriptstyle{#1}/\scriptscriptstyle{#2}}}{\hbox{${#3}$}}}
\newcommand{\tr}[3]{{\hbox{\small{\rm{Tr}}}_{\scriptscriptstyle{#1}/\scriptscriptstyle{#2}}}{\hbox{${#3}$}}}
\newcommand{\cheb}[2]{T_{#1}\big({#2}\big)}
\newcommand{\chsh}[2]{T^\ast_{#1}\big({#2}\big)}

\newcommand{\rtD}{\ensuremath{\sqrt{D}}}
\newcommand{\ta}{\ensuremath{\tau_1}}

\advance\marginparwidth1cm

\begin{document}

\title{SIC-POVMs from Stark units:\\Prime dimensions $n^2+3$\\\ }

\author{Marcus~Appleby}
\email{marcus.appleby@gmail.com}
\affiliation{Centre for Engineered Quantum Systems, School of Physics,
  University of Sydney, Sydney NSW 2006, Australia}

\author{Ingemar~Bengtsson}
\email{ingemar@fysik.su.se}
\affiliation{Stockholms Universitet, AlbaNova, Fysikum, S-106 91 Stockholm, Sweden}

\author{Markus~Grassl}
\email{markus.grassl@ug.edu.pl}
\affiliation{International Centre for Theory of Quantum Technologies,
  University of Gdansk, 80-309 Gdansk, Poland}
\affiliation{\mbox{Max Planck Institute for the Science of Light, 91058 Erlangen, Germany}}

\author{Michael~Harrison}
\email{mch1728@gmail.com}
\noaffiliation

\author{Gary~McConnell}
\email{g.mcconnell@imperial.ac.uk}
\affiliation{\mbox{Controlled Quantum Dynamics Theory Group, Imperial College, London}}

\begin{abstract}
We propose a recipe for constructing a SIC fiducial vector in complex
Hilbert space of dimension of the form $d=n^2+3$, focussing on prime
dimensions $d=p$.  Such structures are shown to exist in thirteen
prime dimensions of this kind, the highest being $p = 19603$.

The real quadratic base field~$K$ (in the standard SIC terminology)
attached to such dimensions has fundamental units~$u_K$ of norm~$-1$.
Let~$\ZZ_K$ denote the ring of integers of~$K$, then~$p\ZZ_K$ splits
into two ideals~$\pp$ and~$\pp'$.  The initial entry of the fiducial
is the square~$\xi^2$ of a geometric scaling factor~$\xi$, which lies
in one of the fields~$K(\sqrt{u_K})$.  Strikingly, the other~$p-1$
entries of the fiducial vector are each the product of~$\xi$ and the
square root of a \emph{Stark unit}.  These Stark units are obtained via
the Stark conjectures from the value at $s=0$ of the first derivatives
of partial $L$-functions attached to the characters of the ray class
group of~$\ZZ_K$ with modulus~$\pp\infty_1$, where~$\infty_1$ is one
of the real places of~$K$.
\end{abstract}

\maketitle

\vskip-2ex
\begin{dedication}
  Dedicated to John Coates.
\end{dedication}

\section{Introduction}\label{sec:intro}
SIC-POVMs, or SICs for short, is a notion from quantum information
theory. Stark units come from number theory. The connection between
them is the point of departure for this paper, so we should first
sketch how this was discovered.

That SICs exist in every finite dimensional complex Hilbert space was
first conjectured by Zauner \cite{Zauner}\nocite{Zauner_English} and
Renes et al. \cite{Renes}. This is relevant to quantum information
theory~\cite{Horodecki}, quantum foundations~\cite{Fuchs13,DeBrota20},
classical information theory~\cite{Fannjiang}, and frame
theory~\cite{Waldron}. For recent reviews see Refs. \citenum{FHS} and
\citenum{Blake}.  The acronym SIC stands for ``symmetric
informationally complete'', while a POVM corresponds to a kind of
measurement.  However, here we will regard SICs more abstractly, as a
geometrical structure.  Mathematically, they are maximal sets of
equiangular lines in complex Hilbert space; or equivalently, maximal
regular simplices in complex projective Hilbert space.  It might seem
that SIC-measurements on systems described by Hilbert spaces whose
dimensions are as high as those that we will encounter in this paper
are completely unrealistic.  However, there is presently much interest
in performing or at least efficiently simulating measurements on high
dimensional systems \cite{Singal}, so such a conclusion may be too
hasty.

Constructing SICs has proved remarkably difficult, and the subject has
proceeded in the spirit of experimental mathematics. Cursory
inspection of the first extensive published catalogue of exact SICs
\cite{Scott} shows that the SIC vectors are built from algebraic
numbers. Close inspection of every case examined reveals\cite{AYAZ}
that these numbers belong to abelian extensions of the {\it real
  quadratic} number field
\begin{equation}
K = \mathbb{Q}(\sqrt{D}),
\end{equation}
 where $D$ is the square-free part of $(d-3)(d+1)$, $d$ is the
dimension, and an extension is called abelian if its
Galois group is abelian.   Moreover every dimension holds a \emph{ray
class} SIC that can be constructed using a {\it ray class field} over
$K$ where the finite \emph{modulus} is $d$ (or $2d$ if $d$ is even)
\cite{AFMY, AFMYpop}.  The minimal modulus defining a particular field
extension is the \emph{conductor} of the extension.  A major
achievement of mathematics in the first half of the twentieth-century
was to classify all abelian extensions of number fields as subfields
of so-called ray class fields.  The Stark conjectures imply (among
other things) that there exists an analytic function from which
generators of these ray class fields can be obtained \cite{stark3,
  stark4}. These generators are known as Stark units. Proving this
conjecture would solve an instance of Hilbert's 12th problem
\cite{Hilbert}. 

That Stark units can be used to construct SICs, at
least in certain cases, is a more recent insight due to Kopp \cite{Kopp}:
indeed, number theorists may be interested in a geometrical problem in which
Stark units play a role.  
To the best of our knowledge, this is the first instance in which complex Stark units appear 
as the solutions of a problem which \emph{prima facie} has nothing to do with number theory. 

Readers who are not number theorists may find it helpful if we
describe the ordinary roots of unity in similar terms. The $n$th roots
of unity are in fact generators of the \emph{cyclotomic} number fields
of modulus $n$, and by the celebrated Kronecker--Weber theorem these
number fields house all the abelian extensions of the rational number
field.  It turns out that the cyclotomic fields
$\mathbb{Q}(e^{\frac{i\pi}{d}})$ are subfields of the particular ray
class fields mentioned in the previous paragraph \cite{AFMY}.  Roots
of unity are needed in the geometrical problem of dividing the circle
into $n$ equal parts. They are needed in the SIC problem too, because
roots of unity are used in the representation of the Weyl--Heisenberg
group, which plays a central role in constructing SICs.

In this paper ``SIC'' means a SIC which is a single projective orbit
with respect to the Weyl--Heisenberg group in a dimension $d >
3$. (For SICs not satisfying these conditions see
Ref.~\citenum{Blake}.)  Thus a SIC is constructed by acting on a {\it
  fiducial} vector $\Psi$ by elements of the Weyl--Heisenberg group.
This means that the SIC is fully specified by the single vector
$\Psi$, on which we accordingly focus. The necessary and sufficient
condition for a $d$ dimensional vector $\Psi$ to be a SIC fiducial
vector is
\begin{equation}
| \langle\Psi|D_{j,k} |\Psi\rangle |^2 =
\begin{cases} 1 \qquad & j=k=0\\ \frac{1}{d+1} \qquad &\text{otherwise}
\end{cases}
\label{eq:sicdefeq}
\end{equation}
for $j$, $k$ running from $0$ to $d-1$.  Here $D_{j,k}$ is a
Weyl--Heisenberg displacement operator in the standard Weyl
representation, as defined in Appendix \ref{sec:AppB}.  We refer to
the complex scalar products $ \langle\Psi|D_{j,k} |\Psi\rangle$ as
{\it SIC overlaps}. The vector $\Psi$ can be reconstructed if all the
overlaps are known, up to an irrelevant phase factor. We define the
\emph{SIC field} to be the number field generated by the ratios of the
components of $\Psi$ together with the root of unity $e^{i\pi/d}$
(which comes in via the operators $D_{j,k}$).

That Stark units appear in the SIC problem was first observed in a
careful study of SIC overlaps \cite{DM}, and this was developed in a
significant way by Kopp \cite{Kopp}. For certain dimensions he used
Stark units to construct the SIC overlaps from which SICs can be
reconstructed. Forthcoming work will present evidence that this
approach will work for every SIC in every dimension \cite{Kopp2}.
Kopp's work served as an inspiration for us, and there are several
points of contact. Our idea is different in that we focus on a special
sequence of dimensions in which we construct a fiducial vector $\Psi$
directly, without going via the overlaps. The reason is that it has
been observed that in some dimensions fiducial vectors can be
constructed using only a subfield of the full ray class field
\cite{Fibonacci, MG, MG2}. Here we will select a special infinite sequence
of dimensions in which we expect that the entire cyclotomic subfield 
of the full ray class field ``decouples'' from a suitably chosen fiducial 
vector.  We then use Stark units in a small ray class field to
construct that fiducial vector. It is not our immediate aim to find a
formula for every SIC; not even in these special dimensions. In this
sense our approach is limited in scope. On the other hand we will be
able to construct SICs in dimensions that are much higher than
achieved by any previous method, whether based on exact solutions or
numerical searches.  The highest dimension reported in this paper is
$d = 19603$ (see Table~\ref{tab:prim}).

The special sequence of dimensions we are interested in consists of
all dimensions of the form $d = n^2+3$, so that $D$ is simply the
square-free part of $d+1$. For these dimensions the standard
conjecture is that the ray class SIC has an anti-unitary symmetry in
addition to the unitary symmetry which appears to be always present
\cite{Scott, Andrew}.  Here \emph{symmetry} refers to a group of
\mbox{(anti-)}linear operators that map the SIC to itself, and fix at
least one of the vectors.  In fact the unitary symmetry in our case is
expected to be of order $3\ell$, where $\ell$ is an odd integer that
is known if $d$ is known (as explained in Section \ref{sec:towers}).
Moreover one can show that all the dimensions in our sequence have a
prime decomposition of the form
\begin{equation}
  d = n^2 + 3 = 2^{e_1}\cdot 3^{e_2}\cdot p_1^{r_1}\cdot \ldots \cdot p_s^{r_s} ,
\end{equation}
where $e_1 \in \{0,2\}$, $e_2 \in \{ 0,1\}$, and all the primes $p_j =
1$ modulo $3$ (see Ref.~\citenum{Anti1}). 
Arguments that we will go into later then
suggest that there exists a fiducial vector 
from which the cyclotomic subfield of the full ray class field decouples. 
The full ray class field will be generated by the components
of the fiducial vector together with the cyclotomic unit that comes
from the Weyl--Heisenberg group.

In this paper we make the further restriction to the case that $d =
n^2 +3 = p$, a prime number. This last restriction is not made for any
fundamental reason, but only because it simplifies the presentation at
some points. It is then natural to ask whether this sub-sequence of
prime dimensions is infinite. `Conjecture F' by Hardy and Littlewood
indeed implies that the sequence is infinite \cite{Hardy, Shanks}, but
this remains unproven.

The equation
\begin{equation}
    d = \left(\sqrt{d+1}+1\right)\left(\sqrt{d+1} -1\right) \label{splitmod}
\end{equation}
together with the fact that $\sqrt{d+1}$ is an algebraic integer in
$K$ means that the prime $d$ splits into two prime ideals. It
follows that the ray class field with finite modulus $d$ contains two
small ray class fields with finite moduli $\sqrt{d+1} \pm 1$ as
subfields. By class field theory, each of them has trivial intersection
with the cyclotomic subfield, which makes them natural candidates for
the small ray class field that is to hold the fiducial vector.

Thus, let $d=n^2 +3 = p$ and define two vectors: an un-normalized
fiducial vector
\begin{equation}
\hat{\Psi} = \begin{pmatrix} x_0 & x_1 & \dots & x_{d-1}\end{pmatrix}^{\rm{T}} , \label{Psi2}
\end{equation}
all of whose components belong to the SIC field, and the SIC fiducial
vector properly speaking which is the unit vector
\begin{equation}
    \Psi= N \hat{\Psi} =  \begin{pmatrix} a_0 & a_1 & \dots & a_{d-1}\end{pmatrix}^{\rm{T}}. \label{Psi1}
\end{equation}
The overall phase of the vector $\Psi$ is arbitrary, but it is
required that the $x_j$ together with the absolute value squared
$|N|^2$ belong to the SIC field.  Because of the symmetries referred
to above, the number of independent components, not counting $x_0$, is
expected to be $(d-1)/3\ell$. (Note that we use the convention that
the denominator in this expression is $3\ell$.)  We now state three
conjectures, and comment on them afterwards:

\begin{conjecture}
[\textbf{Almost flat fiducial}]\label{conj1}
In each dimension in the sequence there exists an un-normalized SIC fiducial
$\hat{\Psi}$ for which
\begin{equation}
    x_j = \begin{cases}
    -2 -\sqrt{d+1}, \qquad & j = 0
    \\
    \sqrt{x_0 e^{i\vartheta_j}}, \qquad & j>0
    \end{cases}
    \label{eq:unnormfid}
\end{equation}
for suitable phases $e^{i\vartheta_j}$ and suitable choices of the
signs of the square roots.  It can be seen that $\hat{\Psi}$
is \emph{almost flat}, in the sense that 
$|x_1| = \dots= |x_{d-1}|$.
\end{conjecture}

\begin{conjecture}
[\textbf{Small ray class fields}] \label{conj2}
The components of $\hat{\Psi}$ generate one of the two small ray class
fields {of $K$ of} finite moduli $\sqrt{d+1}\pm 1$.
\end{conjecture}

\begin{conjecture}
[\textbf{Stark phase units}] \label{conj3}
We now come to the
key hypothesis: the numbers $e^{i\vartheta_j}$ in
eq.~\eqref{eq:unnormfid} are Galois conjugates of Stark units for
one of {these two} small ray class fields. We refer to them as \emph{Stark
phase units}.
\end{conjecture}

To construct a SIC fiducial from these conjectures one needs to order
the components correctly.  Standard conjectures about SIC symmetries
imply that we need $(d-1)/{3\ell}$ phase factors in the almost flat
fiducial.  At the same time, there are $(d-1)/{3\ell}$ Stark phase
units, cyclically ordered by a Galois group.  The action of the
Clifford group \cite{Marcus} on fiducial vectors suggests that one
should order the non-zero indices $j=1,\ldots,d-1$ in
eq.~\eqref{eq:unnormfid} with respect to the multiplicative group of
invertible elements in $\ZZ_d^\times$. For prime dimensions, this
group is cyclic as well.  These orderings can be matched, although we
do not know \emph{a priori} in which way.  Furthermore, there is an
ambiguity in choosing the square roots in \eqref{eq:unnormfid}, which
we are able to reduce to the choice of only one sign.  Thus the
conjectures do not fully specify the vector, and a small amount of
trial and error remains in identifying the correct choices.

Conjecture \ref{conj1} is closely related \cite{Roy,Mahdad} to anti-unitary
symmetry, as will be shown in Section~\ref{sec:IB1}.  Concerning
Conjecture \ref{conj2}, note that for the dimensions under consideration
$\sqrt{d+1}$ is always in $K$, but it is by no means obvious that
$\sqrt{x_0 e^{i\vartheta_j}}$ is in the small ray class
field. Nevertheless this will be made somewhat more plausible in
Section \ref{sec:Gary3}.
Conjecture \ref{conj3} is new, and encourages the hope that one might be able to
use the properties of Stark units to prove SIC existence.  
This paper has a more limited aim. 
Its purpose is to develop techniques
that enable us to calculate fiducial vectors with these properties,
and to verify that they are indeed SIC fiducials.
The conjectures can be modified to cover the case when the
dimension is a composite number of the form $n^2+3$.

Our constructions rely on algebraic number theory, on the
representation theory of the Clifford group, and on advanced
calculational recipes. In order to make the logic clear, also for
readers that are new to one or more of these topics, we now give a
summary of the major features of each section.

Sections \ref{sec:Gary1}--\ref{sec:Gary3} establish the number
theoretical framework for our recipe. Section \ref{sec:Gary1} starts
with a review of the dimension towers $\{ d_\ell \}_{\ell =
  1}^\infty$, introduced in Ref.~\citenum{AFMY}. It is shown that
whenever $d_1$ is of the form $n^2+3$ then so is $d_\ell$ for each odd
integer $\ell$. We go on to prove a number of technical lemmas that
will be needed in the sequel.  Section \ref{sec:Gary2} begins with a
review of some relevant theory of ray class fields and their
moduli. We then consider the splitting of the modulus, as in
eq.~\eqref{splitmod}. At the end of this section we know the degrees
of the resulting small ray class fields. Section {\ref{sec:Gary3}}
introduces the key objects of this paper, the Stark units, and
sketches how they can be calculated as special values of derivatives
of partial Dedekind zeta functions (which are generalisations of the
Riemann and Hurwitz zeta functions). We summarise some key points of Stark's
construction. Finally we prove a non-trivial result regarding the
square root that appears in eq.~\eqref{eq:unnormfid}. At the end of
this section we have in hand the number theoretical underpinning for
the calculations in Section \ref{sec:Markus}.

In Section \ref{sec:IB1} we enter Hilbert space and state the standard
conjectures about SIC symmetries on which we rely.  To do this we need
the representation theory of the Weyl--Heisenberg and Clifford
groups. A short summary of the relevant theory, with references, is
provided in Appendix \ref{sec:AppB}. We show that the standard
conjectures come close to implying that a fiducial vector taking the
form given in eqs.~\eqref{Psi2}--\eqref{eq:unnormfid} exists.  In
Section \ref{sec:IB2} we discuss the ``decoupling'' phenomenon: that
is to say the expectation that a small ray class field should suffice
to construct this fiducial vector.  Indeed, despite their apparent
ubiquity in defining the geometry via the Heisenberg-Weyl group, the
cyclotomic numbers nevertheless seem mysteriously to vanish from the
final vector entries.

In Section \ref{sec:Markus} we state, in ten precise steps, our recipe
for how to construct SIC fiducial vectors from Stark units. We present
a combination of numerical and algebraic techniques that allow us to
carry out the required calculations for dimensions which are
considerably larger than that of any previous example (see
Table~\ref{tab:prim}).  Section \ref{sec:IB3} illustrates these
calculations for two low dimensions ($d = 7$ and $d = 199$) where
computational complexity is not an issue.

In Section \ref{sec:IB4} we compare our method of calculating SICs to
previous methods. We then raise two open questions. Finally we give a
preview of our results in non-prime dimensions, where we have also
found solutions using similar methods; the highest being $d=39604$.
Section \ref{sec:summary} gives our conclusions. At the end of the
paper we will be conjecturally in a position to program a computer to
calculate a SIC from scratch; but not in a position to specify in
advance the precision that will be needed for the numerical part of
the calculations.

Throughout the paper the integer $d \geq 4$ will continue to signify
the {dimension} for the ambient complex Hilbert space in which the
SIC-POVMs, or sets of equiangular lines, are supposed to exist.
Because of their somewhat technical nature, Sections \ref{biquad} and
\ref{glob} may be left out on a first reading.

\section{The geometric dimension $d$ and Pell's equation in $\QQ(\sqrt{D})$}\label{sec:Gary1}
A good reference for this section is the book by Cohn\cite{harv}.  Fix
an integer $d\geq4$.  Write $K$ for the real quadratic
field $\QQ(\sqrt{D})$, where $D\geq2$ is the square-free part
of $(d+1)(d-3)$.  Let $\Gal{K}{\QQ}$ denote the Galois group of the
number field extension $K/\QQ$, with unique non-trivial
element $\tau$.  We shall often use the same symbol, where it will
cause no confusion, for a lifting of $\tau$ to Galois groups of
which $\Gal{K}{\QQ}$ is a quotient.  The field $K$ can be embedded
into the real numbers $\RR$ in one of two ways: either by
sending~`$\sqrt{D}$' to $\sqrt{D}>0$, or to $-\sqrt{D}<0$.  Fix $\jj$
to be the embedding $K \hookrightarrow{} \RR$ under which it is sent
to $\sqrt{D}>0$, with $\jj^{\tau}$ denoting the other.

The unique non-trivial element $\tau \in \Gal{K}{\QQ}$ then
interchanges these embeddings, or \emph{(real) infinite places}, $\jj$
and $\jj^{\tau}$ of $K$.  Note that here we are using the standard
notation from Galois theory, that an object $\alpha$ which is acted
upon by a field automorphism $\sigma$ is written $\alpha^\sigma$, 
with the convention~$\alpha^{(\sigma\tau)} = (\alpha^\tau)^\sigma$.
Once these embeddings are fixed, any non-zero algebraic
number $\alpha \in K$ has a well-defined~\emph{signature}, being one
of the four possibilities $(\pm,\pm)$ given by the sign $\pm$
of $\alpha$ under each respective embedding.

For example, let $D=5$ and consider $u = (1+\sqrt{5})/2 \in K$.
Then $\jj(u) > 0$ but $\jj^{\tau}(u) < 0$, so the signature of $u$
would be $(+,-)$, as for example would be that of $1/u$.  In
contrast $-u$ or $-1/u$ would have signature $(-,+)$; and $u^2$
or $1/u^2$, $(+,+)$.  The algebraic number $u^2$ would then be
referred to as being~\emph{totally positive}; similarly $-u^2$ would
be~\emph{totally negative}.

We shall need the notions~\cite{lang} of the \emph{norm} and the
\emph{trace} of an algebraic number $\alpha \in F$ down to its base
field $k$.  Namely, in our situation where the extension $F/k$ will
always be Galois, the norm and trace of $\alpha$ are respectively
given by the product and sum of all of its Galois
conjugates $\alpha^\sigma$, for $\sigma \in \Gal{F}{k}$: 
$\norm{F}{k}{\alpha} = \prod_{\sigma \in \Gal{F}{k}} \alpha^\sigma$
and $\tr{F}{k}{\alpha} = \sum_{\sigma \in \Gal{F}{k}} \alpha^\sigma$.

By $\ZZ_F$ we will denote the ring of \emph{algebraic integers} of a
number field $F$: that is to say, all elements $\alpha \in F$ whose
minimal polynomial $f_\alpha(X) \in \ZZ[X]$ is \emph{monic}.  The fact
that $\ZZ_F$ is a \emph{Dedekind domain} means in particular that any
ideal factorises uniquely into a product of prime ideals.  The group
of \emph{units}, or multiplicatively invertible elements, of a
ring $R$ will be denoted by $R^\times$.  Note that $\ZZ_F^\times$
consists precisely of those $\alpha \in \ZZ_F$ for which $f_\alpha(X)$
has constant term~$\pm1$.

By Dirichlet's unit theorem, $\ZZ^\times_K$ is isomorphic as an
abelian group to $\{ \pm 1 \} \times u^\ZZ$ for some unit $u$.  The
four possible choices of generator for the infinite part are $\{ \pm u
, \frac{\pm1}{u} \}$, any one of which may be chosen to be
`the' \emph{fundamental unit} of $\ZZ_K$ (or of $K$).  Throughout this
paper we let $u_K \in {\ZZ^\times_K}$ denote the unique choice which
is $>1$ under $\jj$.  Notice that \emph{in the particular
embedding} $\jj$, this gives an ordering of the fundamental units:
$-u_K < -1 < -\frac{1}{u_K} < 0 < \frac{1}{u_K} < 1 < u_K; $ and that
in any event, $\tau$ swaps $u_K$ with $(\norm{K}{\QQ}{u_K}) / u_K$,
which in all of the cases in this paper will be $-1/u_K$ (see
Lemma~\ref{nm1} below).

Before getting into the technical details, let us give a brief
description of what is proven in the next few sections leading up to
Theorem~\ref{skyxi}. Because of Conjecture~\ref{conj3} we are
interested in square roots of Stark units.  
If the standard conjectural properties of Stark units hold (which our
computations indeed verify, see Section~\ref{hypotheses},
Hypothesis~\ref{sqaw}), we know that the Stark units lie in a certain
abelian Galois extension of $K$ and that their square roots will lie
in a further quadratic extension of that which is still abelian over
$K$. It is important for the construction to know that this extension
is always the one obtained by adjoining a particular square root
$\xea=\sqrt{x_0}$, our geometric scaling factor.  We demonstrate this
by showing that the restricted nature of the ramification above primes
over $2$ forces the quadratic extension to be the composite with a
unique quadratic extension of $K$. The restriction can be expressed
either in terms of discriminants or conductors, the two giving
essentially equivalent information for ramified quadratic extensions
at primes over $2$.  We choose to talk primarily in terms of
conductors since this fits naturally into the class field setting.

\subsection{Towers of dimensions over $K$ when $\nkq{u_K} = -1$}\label{sec:towers}
Fix~$D>1$ and~$K = \QQ(\sqrt{D})$ as above. 
As per Ref.~\citenum{AFMY}, we may solve for an infinite series of dimensions~$d_\ell = d_\ell(D)$ for~$\ell = 1,2,3,\ldots$:
\begin{align}\label{deez}
  d_\ell  &=
  \begin{cases}
	u_K^{\ell} + u_K^{-\ell} + 1   & \text{if $u_K$ is totally positive (so $\nkq{u_K} = +1$)}\\
	u_K^{2\ell} + u_K^{-2\ell} + 1 & \text{otherwise ($\nkq{u_K} = -1$).}
  \end{cases}
\end{align}
These are exactly the~$d_\ell\ge 4$ which yield the chosen quadratic field~$K =
\QQ(\sqrt{D})$, via~$D$ being the square-free part of~$(d_\ell-1)^2-4
= (d_\ell+1)(d_\ell-3)$.  For example, again where~$D=5$, we find~$u_K
= (1+\sqrt{5})/2$ and so the sequence of dimensions indexed by
successive powers of~$u_K^2 = (3+\sqrt{5})/2$ is
\begin{equation}
4,8,19,48,124,323,844,2208,5779,15128,39604,\ldots
\end{equation}
Any integer $d = d_\ell(D) \geq 4$ is valid as a dimension and will
correspond to a unique~$D > 1$ and a unique $\ell\geq1$ placing it in the
list of dimensions associated to that $D$.

A fundamental unit $u_K$ is usually found by searching for a solution
to~\emph{Pell's equation} using continued fractions.  Its
norm $\nkq{u_K} = u_K u_K^{\tau}$ may be either $+1$ or $-1$.  To give
an \emph{a priori} condition that the norm should be $-1$ (that is,
given $D$ but without calculating the fundamental unit by some
continued fraction-style algorithm), is one of the oldest open
problems in Diophantine analysis.  However by \eqref{deez}, such
fields are precisely those yielding dimensions of the form $d=n^2+3$,
as follows.

\begin{lemma}[Theorem 1 in Ref.~\citenum{yokoi}]\label{nm1}
Let~$K = \QQ(\sqrt{D})$ be a real quadratic field with fundamental unit~$u_K$. 
With notation as in eq.~\eqref{deez}, the following statements are equivalent: 
\begin{enumerate}[\rm (I)]
\item \label{won}
$\nkq{u_K} = -1$: that is, $u_K$ and $u_K^\tau$ have opposite signs and so the first totally positive power of~$u_K$ is~$u_K^2$. 
\item\label{too}
For every odd positive integer~$\ell$,~$d_\ell$ is of the form~$n_\ell^2+3$ for some integer~$n_\ell$. 
\item\label{tree}
There exists an integer~$\ell$ such that~$d_\ell$ is of the form~$n_\ell^2+3$ for some integer~$n_\ell$. 
\item\label{fore}
$d_1$ is of the form~$n_1^2+3$ for some integer~$n_1$. 
\item\label{fire}
There exists an integer~$n$ such that the square-free part of~$n^2+4$ is~$D$.
\end{enumerate}
\end{lemma}
\begin{corollary}\label{dee}
For any odd~$\ell$, writing $d_\ell = n_\ell^2+3$, the minimal
polynomial of~$u_K^\ell$ is $X^2 - n_\ell X - 1$, and thus $u_K^\ell =
(n_\ell + \sqrt{d_\ell+1})/2$.
\end{corollary}
Note that the trace of $u_K^\ell$ in these cases is $n_\ell = u_K^\ell - u_K^{-\ell}$ rather than $u_K^\ell + u_K^{-\ell}$ as it would be in the totally positive case when~$\norm{K}{\QQ}{u_K} = +1$. 
The even-numbered terms~$d_2,d_4,d_6,\ldots$ in~\eqref{deez} satisfy
the formula (cf. Section~3 of Ref.~\citenum{AFMY}):
\begin{equation}\label{twotimer}
d_{2\ell} = d_\ell(d_\ell-2). 
\end{equation} 
If this product~$d_{2\ell}$ were itself to have the form~$m^2 + 3$ for some integer~$m$, then writing~$d_\ell = n_\ell^2+3$: 
\begin{equation}
m^2 + 3  =  d_{2\ell}  =  d_\ell(d_\ell-2)  =  (n_\ell^2+3)(n_\ell^2+1)  =  n_\ell^4+4n_\ell^2+3  =  n_\ell^2 ( n_\ell ^2 + 4 ) + 3,
\end{equation}
forcing~$n_\ell ^2 + 4$ to be a square, which is impossible for~$n_\ell \geq 1$. 
Hence \emph{every} integer~$>3$ of the form~$n^2+3$ is a~$d_\ell = d_\ell(D)$ for some \emph{odd}~$\ell$ and some square-free~$D\geq2$ for which the real quadratic field~$\QQ(\sqrt{D})$ contains a unit of norm~$-1$. 

On the other hand it is interesting to observe that for \emph{any}
square-free~$D$ and \emph{any}~$k\geq1$, $d_{2k}(D)$ (which upon
rearrangement of~\eqref{twotimer} is $(d_k-1)^2 - 1$), is always of
the form~$d_{2k}(D) = m^2D+3$, where~$m$ is the integer~$(u_K^{\delta
  t} - u_K^{-\delta t})/\sqrt{D}$; defining~$\delta$ in turn as~$2$
when~$\norm{K}{\QQ}u_K = -1$ and~$1$ otherwise.

Hardy and Littlewood made the conjecture (\emph{inter alia})
in~\emph{Conjecture F} of Ref.~\citenum{Hardy} --- which is now
subsumed within the vast so-called \emph{Bateman--Horn Conjecture} ---
that infinitely many integers in the series~$n^2+3$ are in fact prime.
Whether or not this is true, the decomposition of their prime factors
inside~$\ZZ_K$ is determined \emph{a priori}, as follows.

\begin{lemma}\label{spp}
Suppose that~$d$ has the form~$n^2+3$ \rm(\it so~$D \equiv 1,2$ or $5 \bmod 8$\rm)\it. 
\begin{enumerate}[\rm (I)]
\item\label{pea}
If~$p$ is an odd prime dividing into~$d$, then it splits in~$K/\QQ$. 
\item\label{soup} The prime $2$
  \begin{itemize}
    \item splits   iff $D \equiv 1 \bmod 8$ (iff $n_1\equiv0\bmod8$),
    \item ramifies iff $D \equiv 2 \bmod 8$ (iff $n_1\equiv2\bmod4$),
    \item is inert iff $D \equiv 5 \bmod 8$ (iff $n_1$ is odd (and
      so~$2\mid d$);\\
               or  $n_1 \equiv4\bmod8$).
  \end{itemize}
\end{enumerate}
\end{lemma}

\begin{proof}
The assertions about which classes of $D\bmod 8$ appear, and the
relations between $D$ and $n_1$, follow from basic congruence
arguments applied to the definitions.  For (\ref{pea}): an odd
prime $p$ splits in $\QQ(\sqrt{D})/\QQ$ if and only if $\leg{D}{p} =
1$.  Since by definition $(d+1)(d-3) = f^2 D$ for some $f\in\NN$,
when $p\neq3$, we just unwind the definitions using Legendre
symbols:
\begin{equation}
\leg{D}{p} = \leg{f^2}{p} \leg{D}{p} =  \leg{f^2 D}{p}  =  \leg{(d+1)(d-3)}{p}  =  \leg{d+1}{p} \leg{n^2}{p}  =  \leg{d+1}{p}  =  1,
\end{equation}
using the fact that $d \equiv 0 \bmod p$. 
When $p=3$ we simply observe that $D \equiv 1 \bmod 3$. 
Part~(\ref{soup}) follows from the standard result that $\{1,\sqrt{D}\}$
is a $\ZZ$-basis for $\ZZ_K$ iff $D\equiv2,3\bmod4$, see for example
\S~3.7 of Ref.~\citenum{harv}.
\end{proof}
Henceforth in the paper the dimensions will always be of the form $d =
d_\ell = n_\ell^2+3$, where $\ell\geq1$ is odd.  Conversely, given any
positive $n \in \ZZ$ there is associated a unique value of $D$, being
the square-free part of $d+1 = n^2+4$.  Hence by Lemma~\ref{nm1}, for
any fixed $D$ there is a series of odd-numbered (but not necessarily
odd!) dimensions $d_\ell = n_\ell^2+3$ for $\ell=1,3,5,7,\dots$; where
we may recover each $n_\ell$ (up to sign) via
\begin{equation}\label{hellsbells}
n_\ell  = n_\ell(D)  =  \sqrt{d_\ell-3}   =  u_K^\ell - u_K^{-\ell} = \tr{K}{\QQ}{u_K^\ell} . 
\end{equation}

\subsection{The quartic field~$K(\qmuk)$}\label{biquad}
We need to link the geometry of the fiducial vector with the number
theory.  In particular we need an arithmetic handle on the
{scaling} factor which is the square root of the quantity $x_0$
that appears in the fiducial vector, defined in~\eqref{eq:unnormfid}.
We start this by relating the quadratic extension $K(\qmuk) / K$ `at
the bottom of the tower' to the Stark units.  We need to invoke some
concepts from the theory of $p$-adic (local) fields, for which the
book by Cassels\cite{cassLF} is an ideal reference; with more general
information for example in Ref.~\citenum{cassfro} and~\citenum{serreLF}.  
In particular for any prime ideal $\qq$ of the 
ring of integers $\ZZ_K$ of a number field $K$ we shall denote
by $\K_\qq$ and $\ZZ_\qq$ respectively the field of $\qq$-adic numbers
and its ring of integers.

Throughout the paper we shall write $L = K(\qmuk)$.  The minimal polynomial over $K$ for $\qmuk$ is
$X^2+u_K$, with discriminant $-4u_K$.  So $L/K$ must have discriminant
dividing $4\ZZ_K$; in particular $2$ is the only rational prime above
which ramification could occur.  Because the signature of $u_K$
is $(+,-)$, ramification occurs just at the one infinite place $\jj$.
Under any embedding above $\jj$ the field $L$ is complex; whereas
above $\jj^\tau$ it is real.  In particular, therefore, $L/\QQ$ is
never Galois.

By way of background, the polynomial
\begin{equation}\label{quartz}
  X^4 + n_1 X^2 - 1
  = \Bigl(X-\qmuk\Bigr)\Bigl(X+\qmuk\Bigr)\left(X-\frac{1}{\sqrt{u_K}}\right)\left(X+\frac{1}{\sqrt{u_K}}\right)
\end{equation}
generates the quartic extension~$L/\QQ$, by Corollary~\ref{dee}.
Multiplying~\eqref{quartz} by its conjugate, the polynomial which
results from sending~$u_K$ to~$-u_K$ (or equivalently to~$1/u_K$)
in~\eqref{quartz}, gives the octic~$(X^4 + n_1 X^2 - 1)(X^4 - n_1 X^2
- 1) = X^8 - (d_1-1)X^4 + 1$, a generating polynomial for the
splitting field~$\overline{L}=L(i)$, the normal closure of~$L$
with~$\Gal{\overline{L}}{\QQ} \cong D_4$.

Because we only consider extensions where $\ell$ is odd, we are free
to replace every instance of $n_1$ with $n_\ell$ and $u_K$ with
$u_K^\ell$ in the foregoing.  In particular, therefore, we note for
future reference that \emph{for a fixed value of $D$},
\begin{alignat}{9}\label{dagmar}
\textit{The towers $L/K/\QQ$ will be identical for every $d_\ell = d_\ell(D)$,}\qquad\nonumber\\
\textit{for every odd $\ell \geq 1$.}
\end{alignat}
We shall characterise the quadratic extension $L/K$ in terms of restricted ramification
above $\jj$ and the primes of $K$ over 2 in Proposition~\ref{tworam} in the next section.

\section{Calculating the size of the Galois group}\label{sec:Gary2}
Given $\alpha\in \ZZ_K$ we shall often write $(\alpha)$ in the
standard way for the \emph{principal ideal} $\alpha \ZZ_K$ of $\ZZ_K$.
If $\alpha,\eta \in \ZZ_K$ are given, such that the ideals $(\alpha) =
\alpha\ZZ_K$ and $(\eta) = \eta\ZZ_K$ are \emph{co-prime} --- which is
to say, their (unique) factorisations into prime ideals share no
common prime factors; or equivalently the
ideal $\alpha\ZZ_K+\eta\ZZ_K$ is equal to $\ZZ_K$ --- then we
let $\ord{\eta}{\alpha}$ denote the order of $\alpha\bmod (\eta)$.
That is, the order of the cyclic group $\langle \alpha + (\eta)
\rangle$ generated by the image $\alpha + (\eta)$ of $\alpha$
inside $\ZZ_K/(\eta)$.  The next result is a special case of Lemma~12
of Ref.~\citenum{AFMY}.

\begin{lemma}\label{uDorder}
If~$d_\ell$ is of the form~$n_\ell^2+3$, then $\ord{d_\ell}{u_K^2} = 3\ell$. 
\end{lemma}

\begin{corollary}\label{uforder}
If~$d_\ell$ is of the form~$n_\ell^2+3$, then $\ord{d_\ell}{u_K} = 6\ell$. 
\end{corollary}
\begin{proof}
By hypothesis and Lemma~\ref{uDorder}, $\ord{d_\ell}{u_K} = 3\ell$ or $6\ell$. 
Suppose~$\ord{d_\ell}{u_K} = 3\ell$: then since~$u_K^2 + (d_\ell)  \in \langle u_K + (d_\ell)  \rangle$, it follows that~$\langle u_K + (d_\ell)  \rangle  =  \langle u_K^2 + (d_\ell)  \rangle$; moreover we may map one to the other by squaring. 
Now~$\langle u_K^2 + (d_\ell)  \rangle$ is closed under taking norms down to $\left(\ZZ/d_\ell\ZZ\right)^\times$; hence~$\langle u_K + (d_\ell)  \rangle$ is as well. 
But~$\norm{K}{\QQ}{u_K} = -1$ and~$-1\not\equiv+1 \bmod (d_\ell)$, as~$d_\ell\geq4$: hence~$\# \langle u_K + (d_\ell)  \rangle$ must be even. 
So squaring is not a surjective endomorphism of~$\langle u_K + (d_\ell)  \rangle $, giving the desired contradiction. 
\end{proof}

\subsection{The Exact Sequence of Global Class Field Theory}\label{glob}
For a readable treatment of the basic notions from class field theory
which we invoke below, see for example Refs.~\citenum{cohenstev},
\citenum{milne} or \citenum{cassfro}.  For an informal introduction
situating {ray class fields} within the SIC-POVM problem, see
Section~4 of Ref.~\citenum{AFMY}.  For a general characterisation of
\emph{fractional} ideals in Dedekind domains see Chapter~9 of
Ref.~\citenum{atmac}.

Let $\mm_0$ be any \emph{integral ideal} of the ring $\ZZ_K$, and
let $\mm_\infty$ denote some --- possibly empty --- subset of $\{ \jj,
\jj^{\tau} \}$.  The formal product $\mm = \mm_0 \cdot \mm_\infty$ 
is called a \emph{modulus}, by analogy with modular arithmetic over
the integers.  There is a natural partial order on such moduli, based
on unique factorisation into prime ideals in the finite component and
on set inclusion for the infinite part.  In the absence of any
widely-accepted notation we have chosen to write $K^{\mm}$ for
the~\emph{ray class field of $K$ for the modulus $\mm$}, keeping the
alternative notations of subscripts for local fields $F_\qq$,
and $F(\alpha)$ for the extensions of $F$ by an algebraic number or
indeterminate $\alpha$.  The actual construction of a ray class field
is implicit via class field theory and will not be discussed here.

In particular $K^{(1)}$ --- which we shall denote in the standard way
just by $H_K$ --- is the \emph{Hilbert class field} of $K$, where the
$(1) = 1\cdot\ZZ_K$ signifies the ideal $\ZZ_K$ itself, the identity
in the (multiplicative) group $\JJ = \JJ(K)$ of non-zero
\emph{fractional ideals} of $\ZZ_K$.  The \emph{principal (fractional)
ideals} $\mathcal{P}$ of $\ZZ_K$ --- those ideals $(\alpha) =
\alpha\ZZ_K$ which are generated by a single element $\alpha \in K$ ---
form a distinguished subgroup of $\JJ$.  The \emph{ideal class group}
of $\ZZ_K$ is defined to be the quotient $\clK = \JJ / \mathcal{P}$.
By class field theory this is isomorphic to the Galois group
$\Gal{H_K}{K}$ and its order $h_K = \# \clK$ is what is referred to as
the \emph{class number} of $K$ (or of $\ZZ_K$).

An~\emph{abelian} extension $F$ of a number field $L$ is a Galois
field extension $F/L$ for which the Galois group $\Gal{F}{L}$ is
abelian; and a \emph{finite} extension is one which is generated by a
polynomial of finite degree.  The Kronecker--Weber
theorem~\cite{cohenstev} says that any number field which is a finite
abelian extension of $\QQ$ is necessarily a subfield of some
\emph{cyclotomic} field: that is, a field obtained from $\QQ$ by
adjoining a finite number of roots of unity.  Similarly, global class
field theory tells us that any finite abelian extension $E$ of $K$
naturally embeds inside some ray class field $K^{\mm}$.

Let $\mathfrak{R}^{\mm}$ denote the \emph{ray class group
modulo $\mm$}, defined as follows.  Take the subgroup $\JJ^{\mm}$
of $\JJ$ comprised of all non-zero fractional ideals of $\ZZ_K$ which
are co-prime to the finite part $\mm_0$ of the modulus.  Then
define $\mathcal{P}_1^{\mm}$, the \emph{ray group modulo $\mm$}, as
the subgroup of $\mathcal{P}$ of principal fractional ideals $(\alpha)
= \alpha\ZZ_K$ whose generators $\alpha\in K$ are congruent to $1$
modulo $\mm_0$, as well as being \emph{positive} at every real place
in $\mm_\infty$.  We then set $\mathfrak{R}^{\mm} = \JJ^{\mm} /
\mathcal{P}_1^{\mm}$ to be the quotient group.  Note that $\JJ^{(1)} =
\JJ$ and that $\mathfrak{R}^{(1)} = \JJ/\mathcal{P} = \clK$.

Artin's~\emph{reciprocity map} then gives a canonical isomorphism
between $\mathfrak{R}^{\mm}$ and $\Gal{K^{\mm}}{K}$, induced in our
situation by the map taking each prime $\pp$ of $\ZZ_K$ to the
corresponding \emph{Frobenius automorphism} $\sigma_\pp \in
\Gal{K^{\mm}}{K}$.  This will be a crucial ingredient in
Section~\ref{sec:Markus}, enabling us to find a canonical ordering of
the entries of the fiducial vector derived from the Stark units: once
an embedding is fixed from $K$ into $\RR$, the Artin map unequivocally
associates an element of the Galois group $\Gal{K^{\mm}}{K}$ to each
$\ZZ_K$-ideal $\qq$ coprime to $\mm$.

We state the exact sequence of global class field theory (see (2.7) in
Ref.~\citenum{cohenstev}, Theorem 1.7 in Ref.~\citenum{milne}, or \S0
in Ref.~\citenum{tate}).  With \emph{any} number field $K$ as base
field, and \emph{any} modulus $\mm = \mm_0 \cdot \mm_\infty$, the
following sequence (defining the map $\psi$) is exact:
\begin{equation}\label{globcft}
1 \rightarrow U_1^{\mm} \rightarrow {\ZZ^\times_K} \xrightarrow{\psi } \mg{\mm_0} \times{\{\pm1\}}^{\# \mm_\infty } 
	\rightarrow \mathfrak{R}^{\mm} \rightarrow \clK   \rightarrow 1,
\end{equation}
where ${\# \mm_\infty }$ is the number of real primes in $\mm$.  That
is to say, ${\mm_\infty }$ represents the embeddings of $K$ into $\RR$
which are allowed to \emph{ramify} (note that this terminology is not
universal, see for example the alternative notational set-up with the
sets $S$, $T$ in Gras' book~\cite{gras}, or Ref.~\citenum{tate}) in
the ray class field extension: namely, to be extended in such a way
that the extension contains non-real complex numbers.  The term $\ker\psi =
U_1^\mm$ is the subgroup of the global units ${\ZZ^\times_K}$ which
are simultaneously congruent to $1$ modulo $\mm_0$ and positive at the
real places in $\mm_\infty$.  In the case of real quadratic fields
this kernel has $\ZZ$-rank one and is torsion-free except possibly
when the residue class ring $\ag{ \mm_0 }$ has characteristic $2$.

We apply \eqref{globcft} to the situation which will become the normal
context for the remainder of this paper, from Section~\ref{northernxi}
onwards.  Extracting a short exact sequence from the middle
of~\eqref{globcft}, ending with the cokernel of $\psi$, we obtain:
\begin{equation}\label{snort}
1 \rightarrow {\ZZ^\times_K} / U_1^{\mm}  \xrightarrow{\psi}  \mg{\mm_0} \times  {\{\pm1\}}^{\# \mm_\infty}
	\rightarrow  \coker\psi    \rightarrow 1 , 
\end{equation}
where the tail fits back into~\eqref{globcft} via
\begin{equation}\label{snortlet}
1 \rightarrow \coker\psi \rightarrow  \mathfrak{R}^{\mm} \rightarrow \clK   \rightarrow 1 . 
\end{equation}
So $\mathfrak{R}^{\mm}$ is an extension of $\clK$ by $\coker\psi$. 
We need to invoke an independent result from Section~\ref{sec:splitting} just below. 

\begin{lemma}\label{twoprim}
Suppose that $K = \QQ(\sqrt{D})$ is as in the foregoing, with the
additional stipulation that our `dimension' $d_\ell = d_\ell(D)$ is a
\emph{prime} of the form $p = n_\ell^2+3$ for some integer $n_\ell =
n_\ell(D)$ (and where $\ell$, as always, is odd).  By Lemma~\ref{spp}
we know that $p = \pp \pp^\tau$ splits in $\ZZ_K$.  With our notation
$\jj$, $\jj^\tau$ for the (real) infinite places of $K$, and referring
to the notation of the sequence~\eqref{globcft}, the ray class field
$K^{\pp\jj}$ has the property that for a unique cyclic group $\Gamma$
of (odd) order $(p-1)/6\ell$,
\begin{equation}
\coker\psi \cong \Gamma \times C_2 .
\end{equation}
\end{lemma}

\begin{proof}
Under the assumptions of the lemma, the finite part $\pp$ of our
modulus $\mm = \pp\jj$ is a non-rational principal prime ideal whose
norm is an odd rational prime $p \equiv 3 \bmod 4$, and therefore
$\frac{p-1}{2} = \frac{n^2}{2}+1$ is odd and consequently since
$\ag{\pp} \iso \FF_p$,
\begin{equation}
\mg{\pp}  \cong  C_{\frac{p-1}{2}} \times C_2 . 
\end{equation}
By exactness, the image $\im\psi$ of $\psi$ in \eqref{snort} is
isomorphic to $\ZZ_K^\times/U_1^\mm$, where, as we explained in the
introduction to Section~\ref{sec:Gary1}, $\ZZ_K^\times = \langle u_K
\rangle \times \{\pm1\}$.  But the proof of Proposition~\ref{elluva}
below shows that $\ord{\pp}{u_K} = 3\ell$, which is odd, and hence:
\begin{align}\label{imgam}
\im\psi \cong C_{3\ell} \times C_2;
\end{align}
where $C_2$ maps onto the copy of $\{\pm1\}$ in \eqref{snort}
resulting from the single infinite place $\jj$ (that is, $\#\mm_\infty
= 1$); so finally in turn:
\begin{align}\label{cokgam}
\coker\psi \cong C_{\frac{p-1}{3\ell}} . 
\end{align}
But $\frac{p-1}{3\ell}$ equals $2$ times the odd integer
$\frac{p-1}{6\ell}$, so we may define $\Gamma$ as in the lemma.
\end{proof}
We shall revisit this in Lemma~\ref{seetoo} and in the proof of
Theorem~\ref{skyxi}.  There, we will identify the cyclic group
$\Gamma$ with a subgroup of the Galois group of $K^{\pp\jj}/K$.

Having introduced some of the basic exact sequences of global class field theory,
we now prove two results that will be key components in the proof of Theorem~\ref{skyxi}.

\begin{lemma}\label{ramlem}
Let $F$ be a number field and $\pp$ a prime of $F$ above $2$ of
absolute ramification index $e$.  If $u \in F$ is a $\pp$-unit and
$E=F(\sqrt{u})$ a non-trivial quadratic extension of $F$, then the
exponent of $\pp$ in the conductor of $E/F$ is at most $2e$.
\end{lemma}

\begin{proof}
This is a local question which can be handled with local class field
theory (see, e.g., Chapter III of Ref.~\citenum{neukirch}). If $E/F$
is an abelian extension of local fields with conductor $\mathfrak{f}$,
then $\mathfrak{f} = \prod_\qq\mathfrak{f}_\qq$, where
$\mathfrak{f}_\qq$ is the (local) conductor of the abelian extension
of completions $E_\qq/F_\qq$ (see below), where $\qq$ runs over all
places of $F$ in the product and $E_\qq$ is the completion of $E$ at
any of the places of $E$ lying over $\qq$ (Prop.~IV.7.5 in
Ref.~\citenum{neukirch}).  Here, $\mathfrak{f}_\qq$ is a power $\qq^r,
r \ge 0$ of $\qq$, so it is the local (at $\qq$) part of the global
conductor. Thus, we can replace $F$ by the local field $F_\pp$ and $u$
by its image in that, which will be a unit.  If $u$ is a square in
$F_\pp$, then $\pp$ splits in $E/F$ and the exponent of the conductor
is $0$, since locally we get a trivial extension.  So we are reduced
to proving the statement of the lemma for a quadratic extension by the
square root of a unit $u$, of a $2$-adic local field of ramification
index $e$ over $\QQ_2$.  Let~$\PP$ be a prime of~$E$ lying over~$\pp$,
so we are considering~$E_\PP / F_\pp$.

We could prove the local version with an explicit case-by-case
analysis when $e =\text{$1$ or $2$}$, which are the only cases we
shall need later; however it is easier to prove the general case using
local class field theory and some properties of local Hilbert symbols.

Let $(x,y)$ denote the local Hilbert symbol for $m=2$ (see, e.g.,
\S~5, Ch.~III in Ref.~\citenum{neukirch}).  This is a (multiplicative)
bilinear pairing from $F_\pp^\times \times F_\pp^\times$ to
$\langle\pm1\rangle$. The basic properties we need are (Prop.~III.5.2 in Ref.~\citenum{neukirch})
\begin{enumerate}[(i)]
\item \label{itemi} $(x,y)=1$ iff $y \in \hbox{Norm}_{M/F_\pp}(M^\times)$, where $M=F_\pp(\sqrt{x})$.
\item \label{itemii} $(x,y)=(y,x)$.
\end{enumerate}
Notice (\ref{itemii}) can be viewed as a kind of local reciprocity
law, saying that $x$ is a norm from $F_\pp(\sqrt{y})$ iff $y$ is a
norm from $F_\pp(\sqrt{x})$.

As usual, let $U^{(0)} = U$, the units of $F_\pp$, and, for $n \ge 1$,
let $U^{(n)} = \{ u\colon u \in U | u \equiv 1\ \hbox{mod}\ \pp^n\}$.
In local class field theory, the conductor of $E_\PP/F_\pp$ is by
definition $\pp^r$ where $r$ is the smallest integer $\ge 0$ such that
$U^{(r)} \subseteq \hbox{Norm}_{E_\PP/F_\pp}(E_\PP^\times)$ (Def.~III.3.3
in Ref.~\citenum{neukirch}).  So the statement that we have to
prove is equivalent to $U^{(2e)} \subseteq
\hbox{Norm}_{E_\PP/F_\pp}(E_\PP^\times)$.  By the definition of $e$,
$U^{(2e)}$ consists of precisely those units of $F_\pp$ which are
congruent to $1$ modulo $4\ZZ_\pp$.  So equivalently, by
(\ref{itemi}), we have to show that if $v \in U$, $v \equiv 1\bmod
4\ZZ_\pp$, then $(u,v)=1$ which by~(\ref{itemii}) is true if and only
if~$(v,u)=1$.

But (e.g. by Ex.~2.12 of Ref.~\citenum{cassfro}), $v$ is {\it
  $\pp$-primary} and $F_\pp(\sqrt{v})/F_\pp$ is trivial or unramified:
so it has conductor $\pp^0$ (Prop.~III.3.4 in
Ref.~\citenum{neukirch}). Thus, $U \subseteq
\hbox{Norm}_{F_\pp(\sqrt{v})/F_\pp}(F_\pp(\sqrt{v})^\times)$, and for
any unit $w$, $(v,w)=1$.  In particular $(v,u)=1$, which is what we
had to prove.
\end{proof}

\begin{proposition}\label{tworam}
With the notation of Section~\ref{biquad}, if $L_1/K$ is a
subextension of $K^{\mm_0\jj}$ not lying in $H_K$ with $[L_1:K]$ a
power of two, where $\mm_0 = 4\ZZ_K$, then $L_1H_K = LH_K =
H_K(\qmuk)$.
\end{proposition}

\begin{proof}
By the second paragraph of Section~\ref{biquad}, $L/K$ is unramified at all primes
of $K$ except possibly those above $2$, and is ramified at $\jj$ but not at $\jj^\tau$.
Thus by Lemma~\ref{ramlem}, it has conductor dividing $\mm_0\jj$, so 
$L \subseteq K^{\mm_0\jj}$. Since $L/K$ has ramification at $\jj$, $L \not\subseteq H_K$.

The result
will then follow from that fact that $[K^{\mm_0\jj}:H_K]$ equals $2$
times an odd number.  From the exact sequence (\ref{snort}),
$\Gal{K^{\mm_0\jj}}{H_K}$ is isomorphic to the quotient of $G :=
(\ZZ_K/4\ZZ_K)^\times \times \langle\pm1\rangle$ by the image of the
units of $K$, $\langle\pm1\rangle\times\langle u_K\rangle$. We must
show that the two-primary part 
of this quotient group has order $2$. The proof
divides into a number of cases depending on the residue class of
$D\bmod 8$ and the form of $u_K$.  In all cases, under the map from the units into $G$, $-1
\mapsto (*,-1)$ and $u_K \mapsto (*,1)$, since $u_K > 0$ in the
embedding $\jj$ of $K$ into $\RR$.

We let $u_K = (n+m\rtD)/2$ where $n,m$ are positive integers with $n=n_1=\hbox{Tr}_{L/K}(u_K)$. 
Notice $\ZZ_K = \ZZ \oplus \ZZ\ta$ where $\ta = \rtD$ if $D$ is even, and $\ta = (-1+\rtD)/2$ if $D$ is odd.
Note also that $m^2D=n^2+4$ as $\hbox{Norm}_{L/K}(u_K)=-1$. Now, since $D$ is square-free, we can reduce
to a number of cases for $n,m$, enabling us to analyse the individual congruence classes for~$D$ modulo~$8$. 

\begin{enumerate}[(I)]
\item{\label{ONE} \mbox{$4\mid n$.} Let $n=4n_0$. From $D$ square-free, we
  find that $m=2m_0$, with $m_0$ odd, and $m_0^2D=1+4n_0^2$ as well as
  $u_K=2n_0+m_0\rtD=(2n_0+m_0)+2m_0\ta$.
\begin{enumerate}
\item{\label{one}
$n_0$ odd. Congruences modulo $8$ show that $D \equiv 5\bmod 8$
}
\item{\label{two}
$n_0$ even. Here $D \equiv 1\bmod 8$
}
\end{enumerate}
\item{\label{TWO}
\mbox{$2\parallel n$.} Let $n=2n_0$, $n_0$ odd. $D$
square-free $\Rightarrow$\ $m=2m_0$ with $m_0$ odd and $m_0^2D=n_0^2+1$ with
$D \equiv 2\bmod 8$. $u_K=n_0+m_0\rtD$
}
\item{\label{THREE}
$n$ odd. Here $m$ must also be odd and $D \equiv 5\bmod
8$. $u_K=(n_0+m_0)/2+m_0\ta$.
}
}
\end{enumerate}
We now address each possible case which arose in \ref{ONE}, \ref{TWO},
\ref{THREE} for the different residue classes of~$D$ modulo~$8$. 
The notation~$\sfc$ will denote (one of) the prime ideal(s) of~$\ZZ_K$ above~$2$. 

\noindent\bf{(A)} \rm
$D \equiv 5\bmod 8$.
\newline
Here, $(\ZZ_K/4\ZZ_K)^\times = (\ZZ_K/\sfc^2\ZZ_K)^\times$ is of order
$12$ with Sylow $2$-subgroup $\Sigma_2$ consisting of elements $\equiv 1\bmod
2\ZZ_K$ and isomorphic to { the additive group of $\ZZ_K/2\ZZ_K \cong \FF_4$, 
i.e. to the Klein four-group~$V_4$} under $1+2x \mapsfrom x$. The 
quotient $(\ZZ_K/4\ZZ_K)^\times/\Sigma_2$ identifies
with $(\ZZ_K/2\ZZ_K)^\times\cong \FF_4^\times$, a cyclic group of
order $3$.
\newline
The Sylow $2$-subgroup $G_2$ of $G$ is of order $8$ and it is isomorphic to
$\Sigma_2 \times\langle\pm1\rangle \cong C_2 \times C_2 \times C_2$.
Under the identification of $\Sigma_2$ with $\ZZ_K/2\ZZ_K$, $-1 \mapsto 1$
since $-1 \equiv 3\bmod 4$ and $3=1+2\times 1$.
\newline
In case~(\ref{one})
above, $u_K \equiv
1\bmod 2\ZZ_K$ and $u_K \mapsto \ta$ or $\ta+1$ in $\Sigma_2$.
\newline
In case~\ref{THREE}, 
$u_K$ maps to a generator of $(\ZZ_K/4\ZZ_K)^\times/\Sigma_2$
and $u_K^3\mapsto \ta$ or $\ta+1$ in $\Sigma_2$. To see this, note that
$u_K^2=nu_K+1$ which implies that $u_K^3-1 = 2(((n^2+1)/2)u_K+(n-1)/2)$,
the expression in brackets is congruent to $u_K$ or $u_K+1$ modulo $2\ZZ_K$,
and $u_K\equiv\ta \hbox{ or } \ta+1$ modulo $2\ZZ_K$.
\newline
In either case, it follows easily that the image of $\langle
-1,u_K\rangle$ in $(\ZZ_K/4\ZZ_K)^\times$ quotiented by its
Sylow $3$-subgroup is of order $4$, and this gives the result.

\noindent\bf{(B)} \rm
$D \equiv 1\bmod 8$.
\newline
Here, $2$ splits into 2 primes $\sfc$ and $\sfc^\tau$ in $K$. $\rtD\equiv
1$ modulo the square of one of these primes and $-1$ modulo the square of the other.
\newline
$(\ZZ_K/4\ZZ_K)^\times =
(\ZZ_K/\sfc^2\ZZ_K)^\times\times(\ZZ_K/{{\sfc^{\tau}}^2}\ZZ_K)^\times \cong
C_2 \times C_2$, where $(\ZZ_K/\sfc^2\ZZ_K)^\times$ is identified with
$(\ZZ_K/\sfc\ZZ_K) \cong \ZZ/2\ZZ$ under $1+\varpi x \mapsfrom x$, 
{
where
$\varpi$ is any \emph{uniformiser} --- that is, an element of $\ZZ_K$ in $\sfc\backslash\sfc^2$ (e.g. in this case, $2$) --- 
}
and similarly for the $(\ZZ_K/{{\sfc^{\tau}}^2}\ZZ_K)^\times$ factor.
\newline
We are in case~\ref{two} 
above, which shows that $u_K\equiv\pm\rtD$\ mod
$4\ZZ_K$. So we can choose $\sfc$ so that $u_K\equiv 1$\ mod $\sfc^2$ and
$\equiv -1$\ mod ${{\sfc^{\tau}}^2}$. This
translates into $u_K \mapsto (0,1)$ in $(\ZZ_K/4\ZZ_K)^\times$
identified with $\ZZ/2\ZZ\times\ZZ/2\ZZ$ as above.  Clearly $-1
\mapsto (1,1)$. This shows that the image of the units in
$(\ZZ_K/4\ZZ_K)^\times$ or $G$ is again of order 4 as required.
\newline
Note that in this case, if we consider the image of units in just
$(\ZZ_K/\sfc^2\ZZ_K)^\times\times\langle\pm1\rangle$ or
$(\ZZ_K/(\sfc^2)^\tau\ZZ_K)^\times\times\langle\pm1\rangle$, the above
explicit images show that $[K^{\sfc^2\jj}:H_K] = 2$ and $K^{(\sfc^2)^\tau\jj}=H_K$.

\noindent\bf{(C)} \rm
$D \equiv 2\bmod 8$.
\newline
In this case, $G_1 := (\ZZ_K/4\ZZ_K)^\times =
(\ZZ_K/\sfc^4\ZZ_K)^\times$ is of order $8$. To compute structure and
images, we write $\varpi$ for $\rtD$, which is a local uniformiser for $\sfc$, 
and note that the elements of $G_1$ have representatives of the form
$1+a_1\varpi+a_2\varpi^2+a_3\varpi^3$ with $a_i \in \{0,1\}$.  
Since $\varpi^2=2f$ for some $f \equiv 1\bmod 4$, 
we see that $2=\varpi^2$ in $G_1$ and that we can
reduce the product of two elements of $G_1$ back to representative
form in $G_1$ by equating all powers $\varpi^r, r \ge 4$ to $0$ and
replacing a term $b_i\varpi^r$, $b_i \in \ZZ$ by
$0,\varpi^r,\varpi^{r+2},\varpi^r+\varpi^{r+2}$ for $b_i \equiv 0,1,2,3\bmod 4$, respectively.
\newline
A simple computation then shows that $G_1 = \langle v_1\rangle\times\langle
v_2\rangle \cong C_4\times C_2$, where $v_1=1+\varpi$ is of order $4$ and
$v_2=1+\varpi^2$ is order $2$, and further $1+\varpi^3=v_1^2v_2$.
\newline
The image of $-1$ in $G_1$ is just $v_2$. We are in case~\ref{TWO} above,
which shows that $u_K-1$ is in $\sfc\backslash\sfc^2$, since this is true
of $\rtD$ and $n_0-1 \in 2\ZZ$. 
Thus $u_K$ maps to $v_1^rv_2^s$ with $r$
odd and the image of the units in $G_1$ is all of $G_1$. Since the
images of $-1$ and $u_K$ in the full group $G$ are also of orders $2$
and $4$, respectively, the image of the units in $G$ is a subgroup of
index $2$ which is isomorphic to $C_4\times C_2$.
\end{proof}

\begin{remark}\ 
\begin{enumerate}[(i)]
\item
The proof of the last proposition actually
shows that $K^{\mm_0\jj}$ equals $LH_K$, unless $D \equiv 5\bmod 8$ and
$u_K \equiv 1\bmod 2\ZZ_K$ when it is an abelian cubic extension of
$LH_K$.
\item
Further analysis in the various cases shows that
$L/K$ is actually ramified above $2$ with finite part of the conductor $4\ZZ_K$
when $D\equiv 2,5$\ mod $8$ and $\sfc^2$ when $D\equiv 1$\ mod $8$. Furthermore,
in the latter case, $m_0\equiv 1$\ mod $4$, so $\sfc$ is the prime with
$\rtD\equiv 1$\ mod $\sfc^2$.
\item
For the proof of Theorem~\ref{skyxi}, it is not hard to see that we
could get by with the weaker result that the $2$-rank of
$\Gal{K^{\mm_0\jj}}{H_K}$ is $1$ (i.e., its Sylow $2$-subgroup is
cyclic). For this, we could, for example, make use of a more general
result like Theorem~2.2 of Ref.~\citenum{kleine21}, which implies that
the $2$-rank of $\Gal{M}{H_K}$ is $\le 2$, resp. $3$, when $D\equiv
2,5$\ resp. $1$\ mod $8$, where $M$ is the ray class field of
conductor $\jj\jj^\tau$ times a large power of $2$. In the respective
cases, we can easily find $C_2$ and $C_2\times C_2$ subextensions of
$M/H_K$ which are linearly independent of $K^{\mm_0\jj}$ over $H_K$
(with ramification over $\jj^\tau$ and/or higher conductor
ramification over $2$) and this gives the result. However, we think
that the explicit proof of the more precise result is fairly concise
and clearer to the non-specialist.
\end{enumerate}
\end{remark}

\subsection{Splitting of the modulus $(d)$} \label{sec:splitting}
Define $\dd = 1+\sqrt{d+1}$: this is a key invariant. 
We will often just write the symbol $\dd$ for the principal ideal $( \dd ) = \dd\ZZ_K$. 
Notice $\dd^\tau = 1-\sqrt{d+1} = 2-\dd$ and $d\ZZ_K = \dd
{\dd^\tau}$.  
In particular, when $d=p$ is prime, $\dd = \pp$ and ${\dd^\tau} =
\pp^\tau$ are principal prime ideals which split completely in the
Hilbert class field extension $H_K / K$ (for definitions, see
Section~\ref{glob}).  
When it is important to distinguish different levels $\ell$ over the same
field $K$, we shall write $\de = 1+\sqrt{d_\ell+1}$.

Let $d = n^2+3$ be as above. Then $d \equiv 0$ or $3 \bmod 4$ and
moreover if $d$ is even then $d \equiv 4 \bmod 8$.  Hence $d$ is of
the form $4^{s}r$ for some odd $r \in \NN$ and $s \in \{ 0 , 1 \}$.
By Lemma~\ref{spp}, writing $r$ in terms of its distinct prime powers
as $ r = \prod_{q \mid r} q^{t_q}, $ each prime $q$ appearing in the
factorisation of $r$ inside $\ZZ$ splits into co-prime ideals as
$q\ZZ_K = \mathfrak{q}\mathfrak{q}^\tau$.  (In general, when one of
the primes $q$ equals $3$, it is possible to have $\gcd(d,D) = 3$
whence $3$ would be ramified with $\mathfrak{q} = \mathfrak{q}^\tau$.
It is easy to see however that this cannot occur when $d$ has the form
$n^2+3$).  In other words $q^{t_q}\ZZ_K =
\mathfrak{q}^{t_q}(\mathfrak{q}^\tau)^{t_q}$.  We remark that it is
shown in Ref.~\citenum{Anti1} that $q \equiv 1 \bmod 3$ for every one
of these primes $q>3$.

Consider the ideal $\dd = \dd\ZZ_K$: no odd rational prime can divide
it, and its norm is $d\ZZ_K$; so it follows that one but not both of
the factors $\mathfrak{q}^{t_q}$ and $(\mathfrak{q}^\tau)^{t_q}$ must
divide it.  We thus define a unique $\ZZ_K$-ideal factor
$\mathfrak{q}^{t_q}$ dividing into $\dd$ for each distinct rational
odd prime power factor $q^{t_q}$ of $r$.  Let $\mathfrak{r} = \prod_{q
  \mid r}\mathfrak{q}^{t_q}$, and similarly for $\mathfrak{r}^\tau$.
The factor $4$, if it appears, may be regarded for these purposes as
splitting into $2 \times 2$ so finally we see that $\dd =
2^s\mathfrak{r}$.

The individual prime divisors of $\dd$ and ${\dd^\tau}$ are not, in
general, principal ideals themselves, unless $d$ is prime.  For
example in the case of $d=6^2+3 = 39$, each of $3\ZZ_K$ and $13\ZZ_K$
splits into a pair of non-principal ideals; but since the class number
is $2$ the product of any pair of them is principal.

Incidentally, these quantities $1 \pm \sqrt{d_\ell+1}$ form a series
of intermediate \emph{`irrational dimensions'} which arise from
solving eq.~\eqref{twotimer} for $d_\ell$ whenever $\ell$ is odd:
\begin{equation}
0 = X^2 - 2X - d_\ell = (X-1)^2 - (1+d_\ell) .
\end{equation}
Namely, the case of norm $-1$ yields a series ``$d_{\frac{\ell}{2}}$''
corresponding to odd powers of $u_K$, using eq.~\eqref{deez} with
$u_K$ in place of $u_K^2$, viz.:
\begin{equation}\label{halves}
\de  =  1 + \sqrt{d_\ell+1}  =  u_K^\ell + u_K^{-\ell} + 1   =  \text{``$d_{\frac{\ell}{2}}$''}.
\end{equation}
When $\ell$ is even, they align with the series~\eqref{deez} defined in terms of $u_K^2$: so ``$\dd_{2k} = d_k$''. 

We now come to the main result of Section~\ref{sec:Gary2}. 
Write $\varphi \colon \NN \rightarrow \NN$ for the Euler phi-function,
and CRT for the Chinese remainder theorem. 

\begin{proposition}\label{elluva}
  Let $d_\ell\geq7$ be of the form $d_\ell=n_\ell^2+3$ as above.
\begin{enumerate}[\rm (I)]
\item\label{oddo} When $d_\ell$ is odd, $\Gal{K^{\de\jj}}{K}$ has
  order $h_K\varphi(d_\ell)/3\ell$.  In particular therefore in our
  case where $d_\ell=p$ is a prime, $\Gal{K^{\de\jj}}{K}$ has order
  $(p-1)h_K / 3\ell$.
\item\label{eve} When $d_\ell$ is even, so $d_\ell / 4$ is odd,
  $\Gal{K^{\de\jj}}{K}$ has order $h_K\varphi(\frac{d_\ell}{4})/\ell$.
\end{enumerate}
\end{proposition}

\begin{proof}
(\ref{oddo}) Expansion of the divisors $\de,\de^\tau$ above
in terms of the various $\mathfrak{q}$ affords, via the CRT, a split
of $\mg{d_\ell}$ into two isomorphic groups each of order
$\varphi(d_\ell)$:
\begin{equation}
\mg{d_\ell}  \cong  \mg{\de} \times \mg{{\de^\tau}} .
\end{equation}
Define $w = \ord{\de}{u_K}$ and ${w_\tau} = \ord{{\de^\tau}}{u_K}$.
By Corollary~\ref{uforder}, the order of the image of $u_K$ inside
$\mg{d_\ell}$ is $6\ell$, which must therefore be equal to $\lcm(w ,
w_\tau)$.  The exact sequence~\eqref{globcft}, evaluated in our case
where $\mm = \de \jj$ and so ${\# \mm_\infty } = 1$ then shows that
the order of the Galois group $\Gal{K^{\mm}}{K}$ is
$2h_K\varphi(d_\ell)$ divided by the order of the image of the global
units $\langle u_K \rangle \times \langle -1 \rangle$ inside $\mg{\de}
\times \{\pm1\}$, under the map denoted $\psi$ in~\eqref{globcft}.

We chose $u_K$ to have signature $(+,-)$, so since we are only
allowing ramification at $\jj$, its image is $\psi(u_K) =
(u_K+\de,1)$.  The torsion component $\langle -1 \rangle$ injects
diagonally into the right-hand side as the characteristic of the ring
$\ag{\de}$ is not $2$.  So the kernel of $\psi$ is torsion-free and of
rank one.  Indeed $u_K^w$ is a generator; hence the image of $\psi$
has order $2w$.  In summary, we have shown that the Galois group
$\Gal{K^{\de\jj}}{K}$ has order $h_K\varphi(d_\ell)/{w}$.  We must now
show that $w = 3\ell$.

In order to make use of the analogy between~\eqref{deez} and the
Chebyshev polynomials of the first kind $T_{n}$, a shifted version has
been defined in Ref.~\citenum{AFMY}:
\begin{equation}\label{tstar}
\chsh{n}{x} = 1+2\cheb{n}{\tfrac{x-1}{2}},
\end{equation}
which also satisfies the fundamental composition
relation $\chsh{m}{\chsh{n}{x}} = \chsh{mn}{x}$ for every $m,n\geq0$.
The first few polynomials are $\chsh{0}{x}=3$, $\chsh{1}{x}=x$,
$\chsh{2}{x}=x^2-2x$, $\chsh{3}{x}=x^3-3x^2+3$.  The $T^*_{n}$ are
independent of $D$ so given $d_0 = d_0(D) = 3$ and $d_1 = d_1(D)$, we
know all $d_k = d_k(D)=\chsh{k}{d_1}$.  The defining recursion
$\cheb{n}{x} = 2x\cheb{n-1}{x} - \cheb{n-2}{x}$ for the $T_{n}$ yields
for the $ T^*_{n}$:
\begin{equation}\label{shiftycur}
\chsh{n}{x} = x\chsh{n-1}{x} - x\chsh{n-2}{x} + \chsh{n-3}{x}.
\end{equation}
The definition above at~\eqref{deez} of the `dimensions' $d_\ell$ is
in terms of powers of $u_K^2$; but it could equally well be made in
terms of $u_K$ as in eq.~\eqref{halves} above.  Consequently for
any $\ell \geq 1$ (see equation C0 in Ref.~\citenum{AFMY})
\begin{equation}
\chsh{3}{\de} - \chsh{0}{\de}  =  \dd_{3\ell} - 3 \quad\text{is a $\ZZ_K$-multiple of $\de$,}
\end{equation}
hence by \eqref{halves}: 
\begin{equation}\label{fin}
\dd_{3\ell} - 3  =  u_K^{3\ell} + u_K^{-3\ell} + 1  -  3  \equiv  0   \bmod  \de. 
\end{equation}
But since $\ord{{d_\ell}\ZZ_K}{u_K} = 6\ell$, it follows that $u_K^{6\ell} \equiv 1 \bmod \de$ 
and thus in particular, $u_K^{3\ell}  \equiv   u_K^{-3\ell} \bmod \de$. 
(In fact more is true: it is possible using Lemmas~1 and 7 of
Ref.~\citenum{tbang} to show that for every odd $\ell$, $u_K^{3\ell}$
is actually congruent to $1$ or $-1$ mod $\dd_\ell$ (and so
$\dd_\ell^\tau$). 
This is important for non-prime $d$; although we do not need it here).
So, from \eqref{fin}, 
$
2u_K^{3\ell} \equiv 2 \bmod \de .
$
As~2 is invertible modulo $d_\ell$, it follows that $u_K^{3\ell} \equiv 1 \bmod \de$. 

Suppose $w\lneqq3\ell$: since $u_K^{\tau} = -1/u_K$ it follows by taking $\tau$-conjugates of the congruence $u_K^w \equiv 1 \bmod \de$, that $u_K^w \equiv (-1)^w \bmod \de^\tau$
and therefore that $u_K^{2w} \equiv 1 \bmod \de^\tau$. 
On the other hand, ${w_\tau}$ is defined to be minimal with the property that:
\begin{equation}
u_K^{{w_\tau}}  \equiv  1   \bmod \de^\tau. 
\end{equation}
So in fact ${w_\tau} \mid 2w$. 
But $\lcm(w,w_\tau) = 6\ell$, so by our assumption that $w < 3\ell$: 
\begin{equation} 
6\ell  =  \lcm( w , w_\tau )  \leq   \lcm( 2w , w_\tau )  =  2w  <  6\ell,  
\end{equation}
a contradiction. 
Incidentally, that $ w_\tau = 6\ell$ follows immediately from this. 

\noindent(\ref{eve}) This is an easy modification of the argument
for~(\ref{oddo}), using the fact that $D\equiv5\bmod8$ and so $\mg{4}
\cong C_2 \times C_2 \times C_3$. 
\end{proof}

So we have established the promised order of the main Galois group
(the degree of the ray class field extension $K^{\dd_\ell\jj} / K$).
Finally, for future reference, we extend Lemma~\ref{twoprim} as
follows.  For any prime number $p$ and any group $G$ let $\Syl{p}{G}$
denote the Sylow $p$-subgroup of $G$.

\begin{lemma}\label{seetoo}
We work once more with the notation and restrictive hypotheses of
Lemma~\ref{twoprim}: in particular, $d=p$ is a prime, $p\ZZ_K = \pp
\pp^\tau$, and $\pp = \dd\ZZ_K$.

The quadratic subextension of $K^{\pp\jj}/H_K$ --- that is, the field
$\calF_1 = (K^{\pp\jj})^\Gamma$ fixed by $\Gamma$ --- is generated
over $H_K$ by the polynomial $X^2+\dd$, where $\dd=1+\sqrt{p+1}$ as
defined above, with $\dd\dd^\tau = -p$.

Furthermore, the $2$-primary part of the extension in \eqref{snortlet} is \emph{split}. 
In other words, recalling that $\mathfrak{R}^{\pp\jj}  \cong \Gal{K^{\pp\jj}}{K} $: 
\begin{equation}
\Syl{2}{\Gal{K^{\pp\jj}}{K}}   \iso C_2  \times  \Syl{2}{\clK} .
\end{equation}
\end{lemma}

\begin{proof}
Let $E = H_K(\sqrt{-\dd})$.  The discriminant of this extension
divides into $-4\pp$; and it will be complex only above the first
infinite place $\jj$, as $\jj(\sqrt{d+1}) > 1$.  So the first
assertion follows from the fact~\cite{harv0} that primes of $K$ above
$2$ do not ramify in $E/K$; whereas $\pp$ clearly does, so $E$ is
indeed a subfield of $K^{\pp\jj}$ properly containing $H_K$.

The normal closure of $\calF_1$ over $\QQ$ is the compositum of $\calF_1$ and
$(K^{\pp^\tau\jj^\tau})^{\tau\Gamma\tau^{-1}}$.  This compositum is
abelian over $K$ since it is contained in the ray class field of
modulus $p\jj\jj^\tau$; so its Galois group over $K$ is a direct
product of $C_K$ by the {Klein four-group~$V_4$}.  The result follows.
\end{proof}

\section{Stark units and zeta functions}\label{sec:Gary3}
\subsection{Dedekind ray class zeta functions \cite{lang}}
By taking the classical harmonic series $\sum_{n \geq 1} n^{-1}$ and
making it naturally into a function of a complex variable $s$,
initially with real part $\Re(s) > 1$, the \emph{Riemann zeta
function} $ \zeta(s) = \sum_{n \geq 1} n^{-s} $ encodes arithmetic
information about the distribution of the prime integers within $\ZZ$.
This is particularly evident upon rewriting it as an \emph{Euler
product}, with one~\emph{Euler factor} at each prime:
\begin{equation}
\zeta(s)  =  \prod_{p \hbox{\tiny\ prime}} \frac{1}{ 1 - p^{-s} }. 
\end{equation}
Both of these expressions are valid as written for $\Re(s) > 1$;
moreover $\zeta(s)$ has an analytic continuation to the whole complex
plane, which renders it a meromorphic function on $\CC$ with a single
pole of order one at $s=1$. 

By analogy we define a \emph{Dedekind zeta function} attached to any
number field $F$ by taking the natural analogue of a `harmonic series'
for $F$, this time over the non-zero \emph{ideals} of the ring of
integers $\ZZ_F$, and using the same approach to convert it into a
meromorphic function of $s$.  The \emph{norm of an ideal}
$\mathfrak{a}$ of $\ZZ_F$, denoted by $\norm{F}{\QQ}{\mathfrak{a}}$,
is defined to be the cardinality of the finite quotient ring $\ZZ_F /
\mathfrak{a}$.  As is the convention we shall simply write
$\norm{F}{\QQ}{\mathfrak{a}}^z$ for $(\norm{F}{\QQ}{\mathfrak{a}})^z$
for a complex exponent $z$.  The Dedekind zeta function attached to
$F$ is then:
\begin{equation}\label{zFsum}
\zeta_F(s)  =  \sum_{(0) \neq \mathfrak{a}  \hbox{\tiny\ ideals\ of\ }\ZZ_F} \norm{F}{\QQ}{\mathfrak{a}}^{-s}. 
\end{equation}
Once again this expression defines a function for $\Re(s) > 1$ which
may be analytically continued to the whole of $\CC$; it is also
meromorphic with a single simple pole at $s=1$.  By the unique
factorisation into prime ideals we also obtain an expression for
$\zeta_F(s)$ as an Euler product in the region $\Re(s) > 1$:
\begin{equation}\label{zFprod}
	\zeta_F(s)  =  \prod_{\substack{ { (0) \neq  \pp \hbox{\tiny\ prime}}\\{{\hbox{\tiny\ ideals\ of\ }\ZZ_F}}}}
					\frac{1}{ 1  -  \norm{F}{\QQ}{\pp}^{-s} }. 
\end{equation}

Now let $F$ be the ray class field $K^{\mm}$ of our base real
quadratic field $K$ corresponding to the modulus $\mm = \mm_0
\mm_\infty$.  To each element $\sigma \in G = \Gal{K^{\mm}}{K}$ there
corresponds under the inverse of the Artin reciprocity map, an
infinite coset of fractional ideals $\mathfrak{A}_\sigma$ of $\ZZ_K$,
characterised by their multiplicative class modulo the ray group
$\mathcal{P}_1^{\mm}$ defined above.  This is known as the \emph{ray
class} $[\mathfrak{A}_\sigma]$ of each $\mathfrak{A}_\sigma$.  Each
ray class contains a set of prime ideals of Dirichlet density
$\frac{1}{\#G}$.  Let us fix for each such $\sigma$ a representative
integral ideal $\mathfrak{A}_\sigma$.
 
We may then define for each $\sigma \in G$ a partial zeta function:
\begin{equation}\label{parz}
\zeta(s,\sigma)   =   \sum_{
\mathfrak{a}   \in [\mathfrak{A}_\sigma  ]  }\norm{K}{\QQ}{\mathfrak{a}}^{-s}, 
\end{equation}
and then observe from~(\ref{zFsum}),~(\ref{zFprod}) and~(\ref{parz}): 
\begin{equation}\label{zKlass}
\zeta_K(s)  =   \prod_{\pp \mid \mm } 
				\frac{1}{ 1  -  \norm{F}{\QQ}{\pp}^{-s} } \cdot 
					 \sum_{\sigma \in G} \zeta(s,\sigma).
\end{equation}
In other words, once we adjust for the Euler factors from the `bad
primes' $\pp \mid \mm$, we recover the original Dedekind zeta
function.

\begin{remark}\label{hidden}
It is important in what follows to bear in mind that $\mm$ --- and
ipso facto $\ell$ --- is implicit in each $\sigma$, by virtue of
$[\mathfrak{A}_\sigma]$ being a member of a specific ray class group
of that modulus; even though this is not explicit in the abbreviated
notation we use for the `Stark units' $\es$ below.
\end{remark}

\subsection{Stark units as special values of zeta functions}\label{sec:Stark_units}
Aside from References~\citenum{stark3,stark4,tate}, the material in
this section is explained in a down-to-earth manner in our context
in Refs. \citenum{Kopp} and \citenum{roblot}.  

Recall the definition $\dd = 1+\sqrt{d+1}$, so that in particular as ideals of~$\ZZ_K$,~$(\dd)(\dd^\tau) = (d)$. 
From now on our modulus will always be $\mm = \dd \jj$.  Each of the
partial zeta functions \eqref{parz} has an analytic continuation which
is meromorphic on the whole of $\CC$, with just a single simple pole
at $s=1$, the residue of which is independent of $\sigma$.  So the
differences $\zeta(s,\sigma) - \zeta(s,\sigma')$ are entire functions
(cf. Theorem 3.4 in Ref.~\citenum{Kopp}).  The \emph{Stark
units}~\cite{stark3} attached to the extension $K^{\dd\jj}/K$ are
defined from a particular set of these, as we now explain.

We use Stark's symbol $\mathfrak{c}_0$ for a new ideal subgroup of the
ray group $\mathcal{P}_1^{\dd}$, of index $1$ or $2$.
Namely, $\mathfrak{c}_0 \leq \mathcal{P}_1^{\dd}$ consists of the
principal ideals $(\alpha)$ where $\alpha \equiv 1 \bmod \dd$
and $\jj(\alpha) > 0$.  If the index $[ \mathcal{P}_1^{\dd} \colon
  \mathfrak{c}_0]$ is $1$, then the conjectures in
Ref.~\citenum{stark3} are not able to be formulated for the
modulus $\dd\jj$, because in this case~$\zeta(s,\sigma) = \zeta(s,\sigma')$.  
However, as we shall now show, this index is
always $2$ in our situation and so there exists a coset
of $\mathfrak{c}_0$ in $\mathcal{P}_1^{\dd}$, referred to as $T$ in
Ref.~\citenum{stark3} and as $R$ in the equivalent adaptation in
Section 3.2 of Ref.~\citenum{Kopp}, in which each principal
ideal $(\beta) \in \mathcal{P}_1^{\dd}$ satisfies $\jj(\beta) < 0$.

Indeed, as Stark points out, being of index $2$ is equivalent to
requiring that $\jj(u) > 0$ for any unit $u \in \ZZ_K^\times$
satisfying $u \equiv 1 \bmod \dd$.  By our choice of $\jj$, any power
of $u_K$ will in fact be positive under $\jj$.  Since the
characteristic of $\ZZ_K/\dd$ is strictly larger than $2$, we know
from Proposition~\ref{elluva} that the global units which are
congruent to $1$ modulo $\dd = \de$ are just generated by $\langle 1 ,
u_K^{3\ell} \rangle$, and so Stark's condition is fulfilled.
Incidentally, the same line of reasoning also implies that $K^{\dd} =
K^{\dd\jj^{\tau}}$ and $K^{\dd\jj} = K^{\dd\jj\jj^{\tau}}$, because
$\jj^{\tau}(u_K) < 0$ and so the order at the other place is $6\ell$
rather than $3\ell$.

By definition, $T$ is a ray ideal class in its own right, which for
consistency we could denote by $[\mathfrak{t}] \in
\mathfrak{R}^{\dd\jj}$ for some ideal $\mathfrak{t} \in T$.  We shall
denote the automorphism corresponding to this under the Artin map
by $\sigt \in G = \Gal{K^{\dd\jj}}{K}$; this is the non-trivial element
of $\Gal{K^{\dd\jj}}{K^{\dd}}$.  Multiplication by $[\mathfrak{t}] =
[\mathfrak{A}_\sigt]$ represents this same involution
inside $\mathfrak{R}^{\dd\jj}$.  In the first equation on page~65 of
Ref.~\citenum{stark3}, Stark effectively converts this action into a
character which splits the $\zeta(s,\sigma)$ into a positive and a
negative part, yielding a holomorphic difference we denote
by $\delta$:
\begin{equation}
\delta(s,\sigma)  =  \zeta(s,\sigma)  -  \zeta(s,\sigt\sigma). 
\end{equation}
This function can be analytically continued to the entire complex plane, and for every $\sigma \in G = \Gal{K^{\dd\jj}}{K}$ it satisfies: 
\begin{equation}\label{delflip}
\delta(s,\sigma)  =  -\delta(s,\sigt\sigma).
\end{equation}
We finally define the \emph{Stark units} attached to $K^{\dd\jj}/K$ as the quantities
\begin{equation}\label{eq:Stark_units}
\es  =  \exp(\delta'(0,\sigma)). 
\end{equation}
These numbers, as constructed, are \emph{real} numbers of which there are~$\#G$, arranged into $\#G/2$ inverse
pairs $\{ \es,\es^{-1} \}$ by virtue of \eqref{delflip}, see Hypothesis~\ref{nook} below. 
They are conjectured by Stark to be primitive for the extension~$K^{\dd\jj}/K$ 
and so may be viewed as $G$-conjugates of one another within one (any) fixed real embedding
of~$K^{\dd\jj}$. 
Naturally they have exactly the same number~$\#G$ of `conjugates' within the \emph{complex} embeddings 
of~$K^{\dd\jj}$, which lie on the unit circle as explained in Hypothesis~\ref{S1} below. 
These will be the numbers we will need for SIC fiducial vector construction in the succeeding sections.

\subsection{Key features of Stark's units for SIC fiducials}\label{hypotheses}
We summarise some key points of Stark's remarkable construction,
assuming his Conjectures~1 and 2 in Ref.~\citenum{stark3} to be true for
our base field $K$ and modulus $\dd \jj$.  In addition for
part~\ref{sqaw} below we assume Conjecture 1 in Ref.~\citenum{stark4},
together with what seems to be a widely-accepted slight strengthening
thereof, as in Conjecture 4.2 in Ref.~\citenum{tatetok}, for example; 
or for a clearer statement see the `conjecture' in the introductory section of~Ref.~\citenum{roblot2}.

As a dictionary for Stark's notation in Ref.~\citenum{stark3} in terms
of ours: his $k$ is our $K$; his $F$ is our $K^{\dd}$; his modulus
$\ff\pp_\infty^{(2)}$ is our $\dd \jj$; and his $K$ is our
$K^{\dd\jj}$.

\begin{hypothesis}
\label{nook} {\rm (Theorem 1 in Ref.~\citenum{stark3})}
  Let $\Sigma$ denote a set of fundamental units
  for $\ZZ_{K^{\dd}}^\times$.  The set $\Sigma \cup \{\es\colon
  \text{\small$ \sigma \in G$}\}$ generates a set of units
  of $K^{\dd\jj}$ of finite index in $\ZZ_{K^{\dd\jj}}^\times$.
  Notice per the foregoing discussion, that the rank of the additional
  units provided by the Stark units (in pairs) is exactly the
  `deficit' of~$\#G/2$.
 \end{hypothesis}
Note that Theorem 1 in Ref.~\citenum{stark3} applies in our case since
$K^\jj = H_K$ (which follows easily from \eqref{globcft}). Namely, the
characters $\phi$ appearing in Theorem 1, considered as Galois group
characters, are ones on $\Gal{K^{\dd\jj}}{K}$ that don't factor
through $\Gal{K^{\dd}}{K}$. Since $K^\jj = H_K\subseteq K^{\dd}$, all
of these characters must have prime ideal $\dd$ as the finite part of
their conductor.  We shall make use of this in Remark~\ref{BCD} below,
linking Conjecture~\ref{conj4} below to Conjectures~\ref{conj2}
and~\ref{conj3}.

The torsion-free $\ZZ$-rank of the group of units $\ZZ_{K^{\dd}}^\times$
of the maximal totally real subfield ${K^{\dd}}$ of $K^{\dd\jj}$ is
equal to $\#G - 1$ by Dirichlet's theorem, since the degree of the
field extension $K^{\dd} / \QQ$ is $\#G$.  On the other hand the rank
of $\ZZ_{K^{\dd\jj}}^\times$ is equal to $\frac{3\#G}{2} - 1$ as half
of the places are now complex: the extra units therefore have rank
exactly $\#G/2$.  For an explanation in a similar context see
the final part of Section~6 of Ref.~\citenum{AFMY}.

\begin{hypothesis}
\label{S1} {\rm (\S4, p.~74 in Ref.~\citenum{stark3})} The Galois
  element $\sigt$ induces complex conjugation in the complex
  embeddings of $K^{\dd\jj}$.  Since it is also algebraic inversion,
  it forces the Stark units in $K^{\dd\jj}$ to lie on the unit circle
  in their complex embeddings.
\end{hypothesis}
This justifies the
notation in the introduction, whereby the complex-valued Galois
conjugates of the Stark units $\es$ were referred to by the notation
$e^{i\vartheta}$. We will refer to these complex numbers as
\emph{Stark phase units}.

\begin{hypothesis}
\label{tood} 
  In their real embeddings, the $\es$  are all positive.
\end{hypothesis}
This statement holds for the chosen real embedding of $K^{\dd\jj}$
in which the $\es$ as defined above lie, since they are all given by the exponential
of a real value. Since the $\es$ are all $\Gal{K^{\dd\jj}}{K}$ conjugates of each other,
and this Galois group permutes the real embeddings of $K^{\dd\jj}$ transitively, the
positivity extends to all of the real embeddings.

\begin{hypothesis}
\label{sqaw}{\rm (Stark/Tate~\rm{`over-$\ZZ$': Conjecture in Ref.~\citenum{roblot2}}) }
  The extension $K^{\dd\jj}(\qea)$ of $K^{\dd\jj}$ obtained by
  adjoining the square root of any one of the $\es$ is
  itself an abelian extension of $K$.
\end{hypothesis}
As with Hypothesis~\ref{nook} above, we give a dictionary to go from Roblot's statement to our situation. 
Recall our definition of the ideal~$\pp$ by~$p\ZZ_K = \pp\pp^\tau$. 
His set~$S$ of places of~$K$ in our case consists of exactly three elements: 
the two infinite places, denoted by~$S_\infty$ in his notation, together with~$\pp$. 
The distinguished infinite place~$v$ in Roblot is our place~$\jj^\tau$, which as a real place 
`splits completely' in~$K^{\dd\jj}/K$, in the terminology of say Ref.~\citenum{gras}. 

Hence the conditions in part~(3) of Roblot's statement of Stark's conjecture are fulfilled and 
so his~$\varepsilon$ will be a unit. 

Moreover from equation (1) on page~66 of Ref.~\citenum{stark3}, with
$m=1$, the unit~$\varepsilon$ implicitly defined --- via orthogonality
relations for characters --- by statement (1) in the Conjecture in the
introductory section of Ref.~\citenum{roblot2} is equal to one of the
Stark units, up to a root of unity in~$K$ which in our case means up
to sign.  However we may ignore the sign ambiguity since even
extending by~$\sqrt{-1}$ still keeps us in an abelian extension of~$K$
--- that is, the compositum of~$K^{\dd\jj}(\qea)$ and~$K(\sqrt{-1})$
--- and so its subfields will also be abelian.

\subsection{Scaling the Stark phase units}\label{northernxi}
We now need to impose the restriction that our $d_\ell = n_\ell^2+3$
be a prime, and we write $\dd=\pp$ as before.  
We continue to 
assume that the numbers $\qea$ satisfy Hypotheses 1--4 in the previous section.

\subsubsection*{The {geometric} scaling factor $\xea = \sqrt{x_0}$}
As explained in the introduction to Section~\ref{biquad}, we are on a
quest to establish certain number-theoretical properties of the
scaling factor $\xi_\ell=\sqrt{x_0}$ which arises from eq.~\eqref{eq:unnormfid}. 
Notice from the definitions that $\xi_\ell\xi_\ell^\tau = -n_\ell$. 
For $j>0$, the components $x_j$ of the
un-normalized SIC fiducial vector candidate are $x_j=\sqrt{x_0
  e^{i\vartheta_j}}$. The complex numbers $\sqrt{e^{i\vartheta_j}}$ of
modulus one are obtained via the complex embeddings of the square
roots $\qea$ referred to in 
Hypothesis~\ref{sqaw}.  
With
our fixed embedding $\jj \colon K \rightarrow \RR$ under which
$\sqrt{D} > 0$, for $d_\ell = n_\ell^2+3$ for odd $\ell>0$, define a
scaling factor $\xea \in \RR\cdot\sqrt{-1}$ by:
\begin{equation}
\xea  =  \sqrt{x_0}  =   \sqrt{ - 2 - \sqrt{d_\ell + 1} }  .
\end{equation}
The absolute minimal polynomial of $\xea$ is $X^4 + 4X^2 - n_\ell^2$.
Note that $\xea^2 = -(\dd+1)$ and $(\xea^2)^\tau = -(\dd-3)$,
mimicking the $(d+1)(d-3)$ configuration.  Moreover $-\xea^2
(\xea^2)^\tau = d-3 = n^2$.

Recall the notation $L = K(\qmuk)$. 

\begin{lemma}\label{xiuk}
For every odd $\ell\geq1$: $K(\xea)  =  L$; and this is a
quadratic extension of $K$.  Moreover the prime $\pp$ is inert in
$L/K$; while the prime $\pp^\tau$ splits.
\end{lemma}

\begin{proof}
That $K(\qmuk) / K$ is quadratic follows from the definition of $u_K$.
So it suffices to express $\qmuk$ $K$-linearly in terms of $\xea$, for
any odd positive integer $\ell$, using the definitions and
Corollary~\ref{dee}.  When $\ell = 1$, $\qmuk = (1-u_K)\xi_1/n_1$; or
equivalently $\xi_1 = n_1\qmuk/(1-u_K)$.  Replacing every instance of
$u_K$ with $u_K^\ell$ (and so $n_1$ with $n_\ell$) in this formula
gives similar expressions for $\xea$.

The second part follows from the proof of Proposition~\ref{elluva},
bearing in mind that $p \equiv 3 \bmod 4$: $-u_K$ has order $6\ell$
and so is not a square in the cyclic multiplicative group mod $\pp$ of
order $p-1$; conversely it is a square mod $\pp^\tau$ because its
order is $3\ell$.
\end{proof}

\begin{theorem}\label{skyxi}
Fix any $\est$.  Then with notation as above, the field $K^{\pp\jj}(\qea) \subseteq
K^{\pp\jj}(\xea)$ for every odd $\ell$.
\end{theorem}

\begin{proof}
See Remark~\ref{hidden} for the tacit link between $\ell$ and $\sigma$.

Since $K^{\pp\jj}(\qea)$ is Galois (and even abelian) over $K$ by
Hypothesis~\ref{sqaw}, 
Kummer theory tells us that any
two of the $\est$, which are Galois conjugates over $K$, differ
multiplicatively by a square in $K^{\pp\jj}$. Thus, $K_\sigma :=
K^{\pp\jj}(\qea)$ is independent of $\sigma$.

If the $\est$ are squares in $K^{\pp\jj}$ then $K_\sigma = K^{\pp\jj}$
and the theorem is trivial.  So, we assume that $K_\sigma$ is a
quadratic extension of $K$.

We consider the ramification of $K_\sigma/K^{\pp\jj}$. Since $\est$ is
a unit, $K_\sigma/K^{\pp\jj}$ is unramified at all primes of
$K^{\pp\jj}$ which do not lie over $2$ by the same argument as in the
second paragraph of Section~\ref{biquad}.  The real places of
$K^{\pp\jj}$ are those above $\jj^\tau$ and the $\est$ have positive
images under these real embeddings by Hypothesis~\ref{tood}, 
so $K_\sigma/K^{\pp\jj}$ is unramified at these places.

Since $K_\sigma$ is an abelian extension of $K$ containing
$K^{\pp\jj}$, and no finite or infinite places outside of those over
$2$ can ramify in $K_\sigma/K^{\pp\jj}$, the conductor of $K_\sigma/K$
must be $\pp\mm_0\jj$ where $\mm_0$ is some power product of primes of
$K$ over 2. As $K_\sigma \not\subseteq K^{\pp\jj}$, $\mm_0$ is
non-trivial.  If we take the fixed field $F$ of the inertia subgroup
at $\pp$ of $\Gal{K_\sigma}{K}$, then $F/K$ is an abelian
extension of conductor dividing $\mm_0\jj$ with $K_\sigma =
K^{\pp\jj}F$.  As all extensions are abelian, and
$[K_\sigma:K^{\pp\jj}] = 2$, we easily see that we can replace $F$ by
its maximal subfield over $K$ of $2$-power degree over $K$, and so can
assume that $[F:K]$ is a power of $2$.

We now show that $\mm_0\mid 4\ZZ_K$.

\def\gal{\hbox{\rm Gal}}
\newcommand{\gall}[2]{{\mathrm{Gal}^{s}_{#1/#2}}}
Let $\qq$ be any prime of $K$ lying over $2$ and $\QQQ$ any prime of
$K^{\pp\jj}$ lying over $\qq$.  From the fact that $K^{\pp\jj}/K$ is
unramified at $\qq$, the exponent of $\qq$ in $\mm_0$ is the exponent,
$r$ say, of $\QQQ$ in the conductor of the abelian extension
$K_\sigma/K^{\pp\jj}$. To show this, first note that this is a local
question and, as in the proof of Lemma~\ref{ramlem}, we can reduce to looking at the
conductors of $2$-adic fields in the tower of $K_\sigma/K^{\pp\jj}/K$
completed at primes over $\qq$. Denote this completed tower of
fields by $A/B/C$. Then the statement follows from two facts from
local CFT which relate local conductors to higher ramification groups.
The first fact is that, if $B/C$ is unramified, then the upper-numbered
higher ramification groups $\gall{A}{B}$ and $\gall{A}{C}$ coincide
(as subgroups of $\Gal{A}{C}$) for $s \ge 0$. This is trivially true for higher
ramification groups with the lower numbering (Definition~8.1, Chapter III of
Ref.~\citenum{neukirch}). Then the $\eta_{A/B}(s)$, $\eta_{A/C}(s)$ functions as defined on
page 66 of Ref.~\citenum{neukirch} are the same for $s \ge 0$, and as these
define the translation between the upper and lower-numbered groups
(by the definition on page 67 of Ref.~\citenum{neukirch}),
the upper-numbered ramification groups are also the same for $s \ge 0$.
Secondly, the exponent of the conductor in an abelian extension
of local fields is the index of the first trivial higher ramification group
with the upper numbering.
This comes from the definition of the local conductor (see the proof of
Lemma~\ref{ramlem}), the fact that the kernel of the local norm residue
map is the norm group of the top field (see the paragraph after Thm.~2.1, Ch. III
of Ref.~\citenum{neukirch}) and Thm.~8.10, Ch. III of Ref.~\citenum{neukirch}.

Since $\est$ is a unit in $K^{\pp\jj}$, Lemma~\ref{ramlem} shows that $r \le 2e$,
where $e$ is the absolute ramification index of $\QQQ$, which
equals the absolute ramification index of $\qq$. So $r$ does not
exceed the exponent of $\qq$ in $4\ZZ_K$. Since this is true for all
$\qq$, $\mm_0\mid 4\ZZ_K$.

Thus, $F \subseteq K^{\mm_0\jj} \subseteq K^{4\ZZ_K\jj}$ and $F
\not\subseteq H_K$, since $H_K \subseteq K^{\pp\jj}$. By
Proposition~\ref{tworam}, $FH_K=H_K(\qmuk) \Rightarrow
K_\sigma=FK^{\pp\jj}=K^{\pp\jj}(\qmuk)$. The statement then follows
from Lemma~\ref{xiuk}.
\end{proof}

It is a difficult question to determine general conditions for whether
or not the Stark units $\est$ are already squares in their respective
ray class fields, see for example Refs. \citenum{arakawa},
\citenum{roblot2}, or Chapter 10 of Ref.~\citenum{wash}.  Indeed, this
is the question of the `index' referred to in
Hypothesis~\ref{nook}. 
However, based upon observations in the following prime dimensions in
our series:
\begin{alignat*}{5}
& 7, 19, 67, 103, 199, 487, 787, 1447, 2503, 2707, 3847,\\
& 4099, 5779, 8467, 19603, 132499
\end{alignat*}
we venture to make the following conjecture, which would make the
geometry of the subsequent sections of the paper align neatly with the
number theory we have just been studying.  Strictly speaking, this
arithmetic conjecture is necessary for the three main conjectures in
Section~\ref{sec:intro} to hold as stated.  However, should it turn
out not to be true, those three conjectures would simply take on a
somewhat less concise form.

\begin{conjecture}
[\textbf{Stark units are non-squares}]
\label{conj4}
  The Stark unit $\est$ of Theorem~\ref{skyxi} is not a square.  So
  the statement in Theorem~\ref{skyxi} is an equality, viz.:
  \begin{equation}\label{nxi1}
    K^{\pp\jj}(\qea) = K^{\pp\jj}(\xea);
  \end{equation}
  and hence:
  \begin{equation}\label{nxi2}
    \qea\xi_\ell=\sqrt{x_0\es} \in K^{\pp\jj}.
  \end{equation}
\end{conjecture}

The fact that \eqref{nxi2} follows from \eqref{nxi1} is simply that
$\xi_\ell^2$ and $\est$ are numbers whose square roots generate the
same quadratic extension of $K^{\pp\jj}$; hence their ratio is an 
element of $({K^{\pp\jj}}^\times)^2$, i.\,e., a square in $K^{\pp\jj}$.

As we explain in Section~\ref{sec:Markus}, by taking generators of the
(abelian) Galois group we are given a canonical ordering on these
$\qea$.  This is the starting-point for the calculations in this
paper.

\begin{remark}\label{BCD}
As noted in the explanations following Hypothesis~\ref{nook},
Theorem~1 in Ref.~\citenum{stark3} applies in our case. 
Hence the $\Gal{K^{\dd\jj}}{K}$-conjugates of any one of our Stark units
are distinct, which in turn is equivalent to any one of them generating the `small'
ray class field $K^{\dd\jj}$ over $K$.
Conjecture \ref{conj4} says that the components of our fiducial vector all lie in that same field
--- and so generate it --- since their squares give the Stark units up to an element of $K$.
Therefore Conjecture \ref{conj2} is a consequence of Conjectures \ref{conj3} and \ref{conj4} combined. 
\newline
On the other hand, when Stark's conjecture is true and allows our
construction of a fiducial --- i.e. when Conjecture \ref{conj3} is
true --- it is clear that Conjecture \ref{conj2} implies Conjecture
\ref{conj4}, because if the Stark units were squares in $K^{\dd\jj}$,
the fiducial vector components would not lie in that field. Thus,
given Conjecture \ref{conj3}, Conjectures \ref{conj2} and \ref{conj4}
are essentially equivalent.
\end{remark}

\section{SIC components}\label{sec:IB1}

We now go to a Hilbert space of dimension $d = n^2+3$, and make the restriction that $d$ is odd. Our first task is to explain 
why we expect all these Hilbert spaces to contain a SIC fiducial vector of the general form given in eqs. (\ref{Psi1})--(\ref{eq:unnormfid}). 
We will then explain what symmetries this fiducial vector is expected to have, and how a certain cyclic subgroup of the Clifford group is 
expected to permute its components. All our expectations will eventually receive considerable support from the results we report. 

The background we need on the Clifford group is given in Appendix \ref{sec:AppB}. The Clifford group contains a copy of the symplectic 
group $SL(2, \ZZ_d)$ as a factor group.  The unitary representatives of the symplectic group are called symplectic unitaries. In general 
a symplectic unitary transforms a SIC vector into another SIC vector, usually in a different but unitarily equivalent SIC. A symplectic unitary 
that leaves some SIC vectors invariant and permutes the remaining vectors in that SIC among themselves is said to give rise to a symmetry 
of that SIC. An important example of a symplectic unitary (not giving a symmetry) is the Fourier matrix $U_F$, arising as 
\begin{equation}
  F= \left(
  \begin{array}{rr}
    0 & -1 \\
    1 & 0
  \end{array}
  \right)
  \hspace{5mm} \Longrightarrow \hspace{6mm}
  (U_F)_{r,s} =  \frac{\;e^{\frac{2\pi i rs}{d}}\;}{\sqrt{-d}}, \hspace{5mm} r,s \in \{0, \dots , d-1\},\label{eq:Fourier}
\end{equation}
where the global phase factor was chosen so that the matrix elements
belong to the cyclotomic field generated by the $d$th roots of
unity. It can be used to interchange the generators of the
Weyl--Heisenberg group, $U_FXU_F^{-1}= Z$.

For our first task we will rely on a conjecture based on evidence from
numerical searches, saying that when $d = n^2+3$ is odd then it is
possible to find a SIC fiducial vector $\Psi_R$ that is left invariant
by complex conjugation \cite{Andrew}. Complex conjugation is in fact
an anti-unitary operation appearing in our representation of the
extended Clifford group, but for the moment all we need to know is
that $\Psi_R$ is a real vector. Because the operator $X$ is real, this
implies that the overlap $\langle \Psi_R|X|\Psi_R\rangle$ is also
real. Its absolute value is determined by the SIC condition, but we
will make an additional assumption about its sign. The existence of an
almost flat fiducial vector then follows from a known theorem
\cite{Roy, Mahdad} whose proof we will now repeat.

\begin{theorem}\label{litetteorem}
  If there exists a real SIC fiducial vector $\Psi_R$ such that
  \begin{alignat*}{5}
    \sqrt{d+1}\langle \Psi_R|X^i|\Psi_R\rangle = +1
  \end{alignat*}
for all $i = 1, \dots, d-1$ then there exists an almost flat SIC
fiducial vector of the form
\begin{alignat*}{5}
  \Psi_C = N^\prime (\sqrt{|x_0|}, e^{i\alpha_1}, \dots , e^{i\alpha_{d-1}}), 
\end{alignat*}
where $|x_0| = 2 + \sqrt{d-1}$ and $N^\prime$ is a normalizing factor. 
\end{theorem}
\begin{proof}
Denoting the real fiducial vector by $|\Psi_R\rangle$
we define $|\Psi_C\rangle = U_F|\Psi_R\rangle$ and note that
\begin{equation}
  \langle \Psi_R|X^i|\Psi_R\rangle = \langle \Psi_C|U_FX^iU_F^{-1}|\Psi_C\rangle
  = \langle \Psi_C|Z^i|\Psi_C\rangle = \sum_{k=0}^{d-1}\omega^{ik}|a_k|^2,
\end{equation}
where $\{ a_k\}_{k=0}^{d-1}$ are the components of the complex fiducial vector.
It follows that
\begin{equation}
  | a_k|^2 = \frac{1}{d}\sum_{i=0}^{d-1}\omega^{-ki}\langle \Psi_R|X^i|\Psi_R \rangle. \label{Wienerny}
\end{equation}
In effect, keeping the explicit form of the permutation matrix $X$ in mind, we have used a theorem 
relating the autocorrelation in a time series to its power spectrum\cite{Einstein}. By assumption
\begin{equation}
  \sqrt{d+1}\langle \Psi_R|X^i|\Psi_R \rangle
  = \left\{
  \begin{array}{lll}
    \sqrt{d+1} & \text{if $i = 0$,} \\
    1 & \text{if $i \neq 0$.}
  \end{array} \right.
\end{equation}
From equation \eqref{Wienerny} it then follows that
\begin{equation}
  |a_0|^2 = \frac{\sqrt{d+1} + d-1}{d\sqrt{d+1}}, \hspace{8mm}
  |a_1|^2 = \dots = |a_{d-1}|^2 = \frac{\sqrt{d+1}-1}{d\sqrt{d+1}}.
\end{equation}
\noindent Finally we verify that $|a_0|^2/|a_1|^2 = |x_0|$.
\end{proof}

To arrive at the form of the fiducial vector that we postulated in the Introduction we perform a rescaling 
of the components. A short calculation shows that Theorem \ref{litetteorem} implies equation \eqref{Psi2} 
provided we choose
\begin{equation}
  N^2 = \frac{d+3-3\sqrt{d+1}}{d(d-3)\sqrt{d+1}} .
\end{equation}
This choice is dictated by our insistence that the components of the
vector belong to the ray class field, up to an overall factor whose square must
also be in the field.

One more remark is useful. In the notation of Appendix \ref{sec:AppB} there holds
\begin{equation}
  FJF^{-1} = PJ .
\end{equation}
Given that the complex fiducial vector $\Psi_C$ is related to a real vector by a Fourier transformation,
given that $J$ is represented by complex conjugation, and given the representation \eqref{parity}
of the parity operator $P$, it follows that the components of $\Psi_C$ obey
\begin{equation}
  \bar{a}_r = a_{-r}, \label{cc}
\end{equation}
where the bar denotes complex conjugation and the indexing, as always, is modulo $d$. 

We now turn to the symmetries of our SICs. For convenience we restrict the discussion 
to dimensions of the form $d = p$ where $p$ is a prime number equal to 1 modulo 3. This is 
conceptually the simplest case, and indeed we do focus on such prime dimensions in this paper. 

A battle tested conjecture by Zauner \cite{Zauner} implies 
that in every dimension there exists a SIC fiducial vector invariant under a symplectic unitary of
order three. Zauner's conjecture has been refined and extended over
the years \cite{Marcus, Scott, Fibonacci, Andrew}. In particular we
expect that every dimension $d_\ell$ houses a SIC that has a symmetry
of order $3\ell$, where $\ell$ refers to the position of $d_\ell$ in
the dimension towers described in \eqref{deez} of
Section~\ref{sec:towers}. These SICs were called {\it minimal SICs} in
Ref. \citenum{AFMY}, and in this paper our sole concern is with
them. What is special about dimensions equal to $1$ modulo $3$ is that
the conjugacy class of symplectic unitaries that leave some SIC vector
invariant contains a representative that derives from a diagonal
symplectic matrix. This is important because, in the standard Weyl
representation, a diagonal symplectic matrix is represented by a
permutation matrix \cite{Marcus}.  Accordingly, we choose an integer
$\alpha$ such that $\alpha^{3\ell} = 1$, and a symplectic matrix
\begin{equation}
  S = \left(
  \begin{array}{ll}
    \alpha^{-1} & 0 \\
    0 & \alpha
  \end{array}
  \right)
  \hspace{5mm} \Longrightarrow \hspace{5mm}
  (U_S)_{r,s} = \delta_{\alpha r,s}  .
\end{equation}
As advertised the unitary matrix $U_S$ is a permutation matrix, and the by now
standard conjecture says that there exists a SIC fiducial vector $\Psi$
left invariant by $U_S$.

The centralizer of $S$ within the symplectic group is an abelian group generated by a matrix
\begin{equation}
  G = \left(
  \begin{array}{ll}
    \theta^{-1} & 0 \\
    0 & \theta
  \end{array}
  \right)
  \hspace{5mm} \Longrightarrow \hspace{5mm}
  (U_G)_{r,s} = \delta_{\theta r,s},
\end{equation}
where $\theta$ is a generator of the multiplicative group $\ZZ_d^\times$, so that the list $\theta,
\theta^2, \dots , \theta^{p-1} = 1$ runs through all the non-zero integers counted
modulo $p$. In particular
\begin{equation}
\alpha = \theta^\frac{d-1}{3\ell} .
\end{equation}
The unitary $U_G$ is
a permutation matrix leaving the first (or `zeroth') component of our vector invariant,
and permuting the others according to
\begin{equation}
  \Psi_{\theta^j} \rightarrow \Psi_{\theta^{j+1}}.
\end{equation}
Because of the symmetry that we have postulated
\begin{equation}
  \Psi_{\theta^{j+(d-1)/3\ell}} = \Psi_{\alpha \theta^j} = \Psi_{\theta^j}. \label{Zaunersym}
\end{equation}
The SIC fiducials have $d$ components. The first (zeroth) component is $x_0 = \xi^2$ in the notation
used in Section 4.3. The remaining $d-1$ components can be expressed
in terms of a smaller set of independent numbers $z_0, z_1, \dots , z_{(d-1)/(3\ell )-1 }$ by
\begin{equation}
  \hat{\Psi}_{\theta^j \bmod d} = z_{j \bmod (d-1)/3\ell} .
  \label{primetikett}
\end{equation}
(Here we find it convenient to work with the unnormalized vector
$\hat{\Psi}$. See eq.~\eqref{Psi1}.  The point is that the normalizing
factor $N$ does not belong to the ray class field, although its square
does.)  Thus the fiducial vector is made up from $3\ell$ copies of an
ordered set of $(d-1)/3\ell$ complex numbers. The generator of the
symplectic centralizer, $U_G$, permutes these numbers cyclically and
gives rise to Clifford equivalent SICs. 

We now impose both restrictions, so that $d = n^2+3 = p$ where $p$ is a prime number 
(necessarily equal to 1 modulo 3). It is not difficult to see that a real SIC fiducial vector may 
exist in such dimensions. If we represent the symmetry operators with 
permutation matrices the corresponding symplectic matrices are diagonal, and the extended 
centralizer contains an anti-symplectic matrix $J$ represented by pure complex conjugation \cite{Marcus}. 
All the available evidence suggests that a minimal SIC invariant under such an 
anti-unitary symmetry exists if and only if $d = n^2+3$ and $d$ is odd. Theorem 
\ref{litetteorem} then implies that an 
almost flat SIC fiducial vector invariant under the anti-unitary corresponding to the diagonal 
matrix $PJ$ exists as well. The real and the almost flat fiducial 
vectors are in fact left invariant by the same symplectic unitary $U_S$, because they 
are related by the discrete Fourier transformation and there holds that $FSF^{-1} = S^{-1}$. 

To avoid any misunderstanding: At the present 
time the only dimensions for which it
has been proved that a SIC exists are those for which exact solutions have been found.
What we have done in this section is to give a plausibility argument suggesting that
a SIC fiducial vector of the form given in equations \eqref{Psi1}--\eqref{eq:unnormfid} should
exist whenever $d = n^2+3$. Before the end of the paper we will have used this special
form to prove that SICs exist in a number of dimensions where, previously, they were
not even known in numerical form.

\section{Decoupling of the cyclotomic field}\label{sec:IB2}

\noindent The fact that a small subfield of the full ray class field
sometimes suffices to write down a SIC fiducial is not exclusive to
the dimensions we consider here. It was in fact discovered when the
exact solution for $d=323$ was constructed \cite{Fibonacci}, and it
has been noticed in many dimensions since then \cite{MG, MG2}. It is
however especially pronounced\cite{Anti1} when $d = n^2+3$. An
assumption that underlies our recipe is that the fiducial vector in
eqs.~\eqref{Psi1}--\eqref{eq:unnormfid} lies in the small ray
class field $K^\mm$ that was identified in
Section~\ref{sec:splitting}.  This means that the cyclotomic field and
the fiducial vector decouple completely.

We will now explore some consequences of that assumption. The
symplectic group $SL(2,\ZZ_d)$ can be extended to the general linear
group $GL(2,\ZZ_d)$ in a natural way. Consider the Galois
transformation that takes the root of unity
$\zeta_{2d}=e^{\frac{i\pi}{d}}$ to $\zeta_{2d}^k$.  We can represent
this by the $GL$ matrix
\begin{equation}
  H = \left( \begin{array}{cc}
    1 & 0 \\
    0 & k
  \end{array}
  \right) . \label{cycloH}
\end{equation}
This is a natural extension of the definition \eqref{J} for complex
conjugation \cite{AFMY}. It has led to a conjecture, well supported by
evidence, stating that the Galois group that acts on the SIC overlaps
is isomorphic to the quotient of two matrix groups \cite{AYAZ, AFMY,
  Kopp, ACFW, Salamon}. We allow us to state it in the somewhat
simplified form that it assumes in our case (where we deal with ray
class SICs and, since the dimension is a prime, with
centred\cite{AFMY} fiducial vectors of type $F_z$). Then the statement
is that
\begin{equation}
  \mbox{Gal}_{K^{p\jj}/H_K} \cong C/S ,\label{eq:Galois_mat}
\end{equation}
where $K^{p\jj}$ is the field generated by the SIC overlaps, $H_K$ is the Hilbert class field,
$S$ is the symplectic stabilizer of the centred fiducial, and $C$ is
the centralizer of $S$ within $GL(2,\ZZ_d)$. The `overlap field'
$K^{p\jj}$ was denoted by $\mathfrak{R}_1$ when it was first
introduced \cite{AFMY}.  The Galois transformations of the overlaps
take the form
\begin{equation}
  \sigma (\langle \Psi|D_{\bf p}|\Psi \rangle ) = \langle \Psi |D_{G{\bf p}}|\Psi \rangle , \label{C/S}
\end{equation}
where $G \in C/S$ (i.\,e., $G$ is a representative of a coset of $S$
in $C$).

There is another way of looking at this. When $d$ is a prime equal to
$2$ modulo $3$ the non-trivial overlaps form a single orbit under this
Galois group. This is a key ingredient in the work published by Kopp
\cite{Kopp}.  A key ingredient in our work is that for the dimensions
we consider there exists a fiducial for which the stabilizer $S$
consists of diagonal matrices\cite{Marcus}.  It follows that the
centralizer also consists of diagonal symplectic matrices.  This has
two important consequences.  In the first place it means the set of
overlaps of the form $\langle \Psi|X^j|\Psi \rangle$ form an orbit by
themselves. In the second place it means that the Galois group
permutes the components of the fiducial vector, and in the prime
dimensional case considered in this paper it permutes them
cyclically. The fact that the matrix elements of the displacement
operator $X$ are natural numbers (either $1$ or $0$) means that the
overlaps $\langle \Psi|X^j|\Psi \rangle$ are in the same field as the
fiducial components.  The fact that the Galois group cycles through
these comparatively small sets of numbers means that this field is of
much lower degree than the full ray class field---which takes us back
to the point with which we began this section.

Let us change the subject slightly. If we are given a vector in a definite number
field it is possible to check whether it is indeed a SIC fiducial without venturing
outside that number field. This is the content of an interesting observation
\cite{Mahdad, ADF, FHS} that we repeat here because it is not widely known. Let the
components of the vector $\Psi$ be denoted by $\{ a_r\}_{r=0}^{d-1}$. A short calculation
using the representation of the Weyl--Heisenberg group given in
Appendix \ref{sec:AppB} suffices to verify that
\begin{equation}
  G(i,k) \equiv \frac{1}{d}\sum_j\omega^{kj}|\langle \Psi |X^iZ^j|\Psi \rangle |^2
  = \sum_{r=0}^{d-1} \bar{a}_r\bar{a}_{r+k-i} a_{r-i}a_{r+k} ,
\end{equation}
where $X$ and $Z$ are the group generators. Thus, if
\begin{equation}
  G(i,k) = \sum_{r=0}^{d-1} \bar{a}_{r+i}\bar{a}_{r+k} a_{r}a_{r+i + k} , \label{eq:Gik1}
\end{equation}
then it follows from the invertibility of the discrete Fourier transform that
\begin{equation}
  |\langle \Psi |X^iZ^j|\Psi \rangle |^2 = \frac{d\delta_{i,0}\delta_{j,0} + 1}{d+1}
  \quad \Longleftrightarrow \quad
  G(i,k) = \frac{\delta_{i,0} + \delta_{k,0}}{d+1} . \label{eq:SIC_conditions2}
\end{equation}
Every reference to the roots of unity has been made to disappear.

For the form of the fiducial vector that we are using it 
is the case that $\bar{a}_r = a_{-r}$, with
labels counted modulo $d$ (see eq.~\eqref{cc}). Hence the conditions become
\begin{equation}
  G(i,k) = \sum_{r=0}^{d-1} a_{-r-i}a_{-r-k} a_{r}a_{r+i + k} = 0 , \quad 0 < i \leq k < d . \label{Gik}
\end{equation}
This defines an algebraic variety over $\CC$. By changing the
range of the summation index we find that
\begin{equation}
  G(i,k) = G(i,-k) = G(-i,k) = G(-i,-k) .
\end{equation}
Moreover the conditions on $G(0,i)$ and $G(i,0)$ are automatically
obeyed for our almost flat fiducial vectors.
This allows us to decrease the number of equations to be checked. More may be true
\cite{FHS}. Be that as it may, the alternative form \eqref{Gik} of the SIC conditions
offers a significant computational speed-up when checking the SIC property, in particular
for exact solutions.

\section{Our recipe}\label{sec:Markus}
\noindent With this preparation we are ready to introduce our recipe
for how to construct a SIC in a prime dimension of the form $d =
n^2+3$.  As before, $D$ is the square free part of $(d+1)(d-3)$ and
$K=\QQ(\sqrt{D})$. The Hilbert class field over $K$ is denoted by
$H_K$, and its degree over $K$ equals the class number $h=h_K$.  The
degree of the ray class field $K^{\pp\jj}$ over the Hilbert class field
equals $m=(d-1)/3\ell$, and the Galois group $K^{\pp\jj}/H_K$ is cyclic
of order $m$.

Recall that the entries of the fiducial vector are complex numbers
which are, with the exception of $x_0$, proportional to square roots
of Stark units with respect to their complex embeddings.  By the Stark
conjectures, special values of zeta functions provide us a way to
compute real approximations of Stark units, i.\,e, with respect to
their real embeddings.  In order to directly switch between the real
and the complex embeddings, we would need exact, algebraic
expressions.  For this, we would first have to compute the exact
minimal polynomial of the Stark units to define the corresponding
number field.  Then we would have to compute the roots of the defining
polynomial in that number field.  While there is no theoretical
obstacle, these calculations are only feasible for small examples such
as those in Section \ref{sec:IB3}.  When the dimension and the degree
of the ray class fields grow, those direct calculations quickly become
infeasible.  Instead, we are using a combination of numerical and
algebraic techniques that allow us to carry out the calculations for
much larger cases.  Additional refinements might allow to push the
computational limits even further.

In all numerical steps, we fix the embedding $\jj\colon K\rightarrow
\RR$ with $\jj(\sqrt{D})>0$.
Our recipe to compute a fiducial vector in prime dimension $d=n^2+3$
{may then be} summarized as follows:
\begin{enumerate}
  \item compute the sequence of numerical Stark units and their exact
    minimal polynomial $p_1(t)\in K[t]$
  \item compute the exact Galois polynomial $g_1(t)\in K[t]$ corresponding
    to an automorphism $\sigma_m$ of order $m$ that fixes the Hilbert
    class field
  \item apply the automorphism $\tau\colon\sqrt{D}\mapsto-\sqrt{D}$
    to obtain the polynomials $p_2(t)=p_1(t)^\tau$ and
    $g_2(t)=g_1(t)^\tau$
  \item {factorise} the polynomial $p_2(t)$ over the Hilbert class field
    and pick a factor $p_3(t)$
  \item compute numerical approximations $\tilde{y}_j$ of the roots {$y_j$} of
    $p_3(t)$ and order them using {repeated iterations of} the Galois polynomial $g_2(t)$
  \item {factorise} the polynomial $p_3(t^2/x_0)$ over the Hilbert class
    field and pick a factor $p_4(t)$
  \item compute the square roots $z_j=\pm\sqrt{x_0 y_j}$ and choose the
    sign such that $p_4(z_j)=0$
  \item search for a primitive element $\theta$ of $\ZZ_d$ together
   with a global sign for the square roots such
    that eq.~\eqref{primetikett} yields a fiducial vector
  \item compute the exact Galois polynomial $g_4(t)\in H_K[t]$ from
    the numerical square roots $z_j$ and compute an exact fiducial
    vector over the field
    $\LL\cong H_K[t]/(p_4(t))$\label{step:exact_fiducial}
  \item check the SIC-POVM conditions \eqref{eq:SIC_conditions2}.
\end{enumerate}
The roots $y_j$ are in fact the sought-after Stark phase units, $y_j =
e^{i\vartheta_j}$, but we use a different notation because here we
treat them strictly as algebraic numbers. Clearly, step
\eqref{step:exact_fiducial} can be omitted if one is only interested
in a numerical solution.

In the following, we will discuss the steps of this recipe in some
detail.

\subsection{Numerical Stark units and their exact minimal polynomial}

\noindent The first and currently most time-consuming step in our
recipe is the numerical calculation of the Stark units to sufficient
precision.  As in Section \ref{sec:splitting}, let $K^{\pp\jj}/K$
denote the ray class field over $K=\QQ(\sqrt{D})$ and
$G=\Gal{K^{\pp\jj}}{K}$. For every automorphism $\sigma\in G$, the
corresponding Stark unit \eqref{eq:Stark_units} is given by
\begin{equation}
  \epsilon_\sigma=\exp(\delta'(0,\sigma)),
\end{equation}
where $\delta'(s,\sigma)$ is the first derivative with respect to $s$
of the difference of partial zeta functions
\begin{equation}
  \delta(s,\sigma)=\zeta(s,\sigma)-\zeta(s,\sigma_T\sigma)
  \label{eq:zeta_diff}
\end{equation}
In order to compute the values $\delta'(0,\sigma)$, we are using
built-in functions {in} Magma \cite{Magma} (or PARI/GP \cite{PARI2})
for Hecke $L$-functions (see also Definition 3.6 in
Ref.~\citenum{Kopp}).  The partial zeta functions $\delta(s,\sigma)$
are obtained via finite Fourier transformation with respect to the
abelian group $G$ (i.\,e., using the orthogonal characters of $G$)
from the Hecke $L$-functions.  There is a sign symmetry of the partial
zeta function with respect of the action of $\sigma_T$ which implies
that half of the first derivatives of the $L$-functions vanish.
Moreover, the fact that the derivatives of the zeta function are real
implies that the derivatives of the $L$-functions come in complex
conjugate pairs (or have real values). Using those symmetries, it
suffices to compute only about one quarter of the derivatives of the
$L$-functions.  We have implemented functions in Magma that make use
of those symmetries.  First, we compute the Stark units with, say,
$20$ digits {of precision, allowing us} to identify the complex
conjugate pairs.  Zeros of the derivative of the $L$-functions can be
easily {calculated} from the corresponding Hecke character.

The numerical approximation of the minimal polynomial of the Stark
units is given as
\begin{equation}
  \tilde{p}_1(t)=\prod_{\sigma\in G}(t-\tilde{\epsilon}_\sigma)=\sum_{i=0}^{|G|} \tilde{c}_i t^i,\label{eq:minpoly}
\end{equation}
where the tilde indicates numerical approximation.  The coefficients
$c_i$ of the exact minimal polynomial are integers in the field
$K=\QQ(\sqrt{D})$. They can be expressed as $c_i = c_{i,0}+c_{i,1}
u_K$ with $c_{i,j}\in\ZZ$, where $u_K$ is a fundamental unit of $K$.
In order to obtain the coefficients $c_{i,j}$, we apply an integer
relation algorithm to $(1,u_K,\tilde{c}_i)$.  More precisely, we fix
an embedding $\jj\colon K\rightarrow \RR$ and consider
$(1,\mathfrak{j}(u_K),\tilde{c}_i)$.  Note that \emph{a priori}, we do
not know which precision of the numerical Stark units is sufficient to
find the correct coefficients $c_{i,j}$.  It turns out that the
absolute error of the coefficients $\tilde{c}_i$ is approximately
bounded by the absolute error of the numerical Stark units
$\tilde{\epsilon}_\sigma$ multiplied by the maximal absolute value of
the coefficients $\tilde{c}_i$ (see Theorem 2.12 in
Ref.~\citenum{ReIp11}). 

Hence the required precision appears to be at
least twice the size of the largest coefficient, which is the height
of the polynomial.  Moreover, as we are looking for an integer
relation with three components, one would expect that the required
precision {would have} to be tripled.  Initially, we calculated the
numerical Stark units with higher precision and checked whether the
integer relations were stable when reducing the precision.  Luckily,
we {observed} that a precision {requirement} of a bit more than twice the size of the
largest coefficient $\tilde{c}_i$ was sufficient for our examples.

This aspect is important, as for high precision the run-time of the
algorithm we use to compute the numerical Stark units appears to scale
like $(\text{\#digits})^{3.3}$.  Hence doubling the precision
increases the run-time by about a factor of ten.  In Table
\ref{tab:timings_$L$-function} we indicate the CPU time for computing
the numerical derivatives of the $L$-functions.  For larger examples,
we used separate processes for the individual values.  This
allowed us, for example, to compute the numerical Stark units for
$d=2707$ in about one calendar month instead of about $2.5$ years.
Eventually, we used the implementation of $L$-functions in
PARI/GP \cite{PARI2}, which appears to be more efficient than our
implementation in Magma and {which moreover} allows the use of multiple CPU cores
on {an} HPC cluster.

\begin{table}[hbt]
  \caption{Run-time to compute the derivatives of the $L$-functions at
    $0$. Note that we have used different computers and different
    versions of our Magma programme so that the values are just
    indicators for the complexity. For the last three dimensions we have
    additionally used PARI/GP.}
  \medskip
  
  \def\arraystretch{1.1}\tabcolsep2\tabcolsep
  \begin{tabular}{|r|rr|rr|r|r@{ }l|}
    \hline
    \multicolumn{1}{|c|}{$d$} &
    \multicolumn{2}{c|}{$\deg(K^{\pp\jj}/K)$} &
    \multicolumn{2}{c|}{log height} &
    \multicolumn{1}{c|}{precision} &
    \multicolumn{2}{c|}{CPU time}\\
    \hline
    &&&&&&&\\[-3ex]
    \hline
      $487$ &\kern2em $324$ &&\kern1.8em  $424$ && $1000$ digits & $251$ & hours\\
      $787$ &  $262$ &&  $299$ && $1000$ digits & $118$ & hours\\
     $2707$ &  $902$ && $1861$ && $3800$ digits & $900$ & days\\
     $4099$ & $1366$ &&  $974$ && $2000$ digits & $170$ & days\\
     $5779$ &  $214$ &&  $127$ &&  $300$ digits &  $18$ & min\\[0.5ex]
     \hline
     $1447$ &  $964$ && $2158$ && $4600$ digits & $111$ & days\\
     $2503$ & $3336$ && $6464$ && $13100$ digits & $60.5$ & years\\
    $19603$ & $2178$ && $1754$ && $4000$ digits &  $82$ & days\\
    \hline
  \end{tabular}
  \label{tab:timings_$L$-function}
\end{table}

\subsection{Exact Galois polynomial of the Stark units}

\noindent The numerical Stark units $\tilde{\epsilon}_\sigma$ are
indexed by the elements of the Galois group $G$ of $K^{\pp\jj}/K$. Hence we
know how the Galois group permutes them.  Moreover, our implementation
in Magma allows us to identify an automorphism $\sigma_m$ of order $m$
that fixes the Hilbert class field. We choose an indexing of the $hm$
Stark units $\epsilon_\sigma=\epsilon_j$ such that the action of
$\sigma_m$ on them is given by the permutation
\begin{equation}
  \pi_m =
  \bigl(0\;1\,\ldots\,m-1\bigr)
  \bigl(m\;m+1\,\ldots\,2m-1\bigr)\ldots
  \bigl(m(h-1)\;m(h-1)+1\,\ldots\,mh-1\bigr)
\end{equation}
consisting of $h$ cycles of length $m$ each.  With this indexing,
\begin{equation}
  \epsilon_j^{\sigma_m} = \epsilon_{j^{\pi_m}}.\label{eq:action_Stark}
\end{equation}
The mapping of \eqref{eq:action_Stark} can also be realized as a
polynomial function, i.\,e., we are looking for a polynomial $g_1(t)$
such that
\begin{equation}
  g_1(\epsilon_j) =\epsilon_{j^{\pi_m}}
    \qquad\text{for $j=0,\ldots,mh-1$.}\label{eq:interpolation}
\end{equation}
We refer to the unique such polynomial $g_1(t)$ of degree at most $mh-1$ as
Galois polynomial.  The conditions \eqref{eq:interpolation} are
invariant under the action of the whole Galois group $G$.  Therefore,
$g_1(t)$ is a polynomial with coefficients in $K$, which is fixed by
$G$.  Note that a similar approach is discussed in Ref.~\citenum{Kopp}
where the coefficients of the Galois polynomial are computed by
solving a linear system with a Vandermonde matrix.  Using standard
algorithms, this has cubic complexity in the number $mh$ of Stark
units. We compute a numerical approximation $\tilde{g_1}(t)$ via
polynomial interpolation, which has quadratic complexity in the
implementation in Magma. There are more sophisticated
methods \cite{Kung} using fast Fourier transformations with complexity
$O(mh \log^2(mh))$ which we have not implemented.

The coefficients of the exact polynomial $g_1(t)\in K[t]$ are again
determined via an integer relation algorithm.  Unlike the situation of
the exact minimal polynomial of the numerical Stark units, we do not
yet have a good heuristic for the required precision.  This is not
necessarily a severe problem, as we can increase the precision of the
numerical Stark units using Newton's method, which has quadratic
convergence.  For Newton's method, we use numerical approximations of
the exact polynomial $p_1(t)$ and its derivative to arbitrary
precision.  However, it turns out that, in general, the coefficients
of $g_1(t)$ have very large numerators/denominators.  For $d=2707$, we
used a precision of about two million digits.  Computing
$\tilde{g}_1(t)$ of degree $901$ using numerical polynomial
interpolation took about $18.5$ hours.  Determining the $902$ exact
coefficients took about $420$ CPU hours in total, which was done in
parallel on multiple CPU cores. For dimension $d=19603$, the
coefficients of the exact Galois polynomial have denominators that are
integers with more than one million digits.

In the end this step has turned out to be a potential computational
bottleneck.  For dimension $d=2503$, we have been able to compute the
(conjectural) exact minimal polynomial $p_1(t)$ of degree $3336$ using
a precision of $13100$ digits.  This took in total about $60.5$ CPU
years on the HPC cluster.  Using a precision of ten million digits
appears to be insufficient to obtain the exact Galois polynomial
$p_1(t)$. Since the initial version of this manuscript, however, we
have developed an alternative approach that allowed us to compute a
numerical fiducial vector for dimension $2503$.  We will discuss this
approach elsewhere when addressing composite dimensions $n^2+3$.

\subsection{Flipping the sign}
\noindent The calculation of the exact minimal polynomial $p_1(t)$ of
the Stark units starts with their real approximations, corresponding
to the real embeddings of the ray class field $K^{\pp\jj}$.  For the
fiducial vector we need the complex Stark phase units, i.\,e., their
values on the unit circle in the complex embeddings.  The embeddings
of $K^{\pp\jj}$ are defined by the roots of the absolute minimal
polynomial of the Stark units over $\QQ$, which lie in the normal
closure of $K^{\pp\jj}$.  Instead of computing the exact roots in the
extension field, we apply the automorphism
$\tau\colon\sqrt{D}\mapsto-\sqrt{D}$ to the coefficients of the exact
minimal polynomial $p_1(t)$ to obtain the polynomial
$p_2(t)=p_1(t)^\tau$. Note that the polynomial $p_1(t)p_2(t)\in\QQ[t]$
is the absolute minimal polynomial of the Stark units over $\QQ$.
When computing the coefficients of the exact polynomial $p_1(t)$, we
have used the embedding $\jj$.  Now, with respect to the same
embedding, the roots of the polynomial $p_2(t)$ will be the desired
complex Stark phase units.  The exact Galois polynomial
$g_2(t)=g_1(t)^\tau$ has the property that it permutes the Stark phase
units in the same way as the polynomial $g_1(t)$ permutes the real
Stark units.

\subsection{Minimal polynomial over the Hilbert class field}

\noindent When the class number $h$ is larger than one, the Hilbert
class field is a proper extension of degree $h$ of $K$. Then the exact
minimal polynomial $p_2(t)\in K[t]$ factorises as
\begin{equation}
  p_2(t)=p_2^{(1)}(t)\cdot p_2^{(2)}(t)\cdot\ldots\cdot p_2^{(h)}(t),
\end{equation}
with $h$ factors $p_2^{(j)}\in H_K[t]$ of degree $m$ each. In our
examples, the class number is small, and hence exact factorisation of
$p_2(t)$ over the Hilbert class field does not take much time.  We
pick any of the factors, e.g., $p_3(t)=p_2^{(1)}(t)$. Other factors
yield fiducial vectors that are related by a Galois automorphism
fixing $K$, but acting non-trivially on $H_K$.  When $h=1$, we have
$p_3(t)=p_2(t)$.

\subsection{Numerical Stark phase units}
\noindent
The next step is to compute $m$ numerical Stark phase units
$\tilde{y}_j$ as roots of the exact polynomial $p_3(t)$ with respect
to the fixed embedding $\jj$.  At the same time, we want to order them
according to the action of the Galois automorphism $\sigma_m$.  For
this, we compute one complex root $\tilde{y}_0$ to moderate precision,
which again can be increased using Newton's method. The other roots
are computed using the exact Galois polynomial $g_2(t)$ via
\begin{equation}
  \tilde{y}_{j+1}=g_2(\tilde{y}_j).\label{eq:galois_cycle}
\end{equation}
To compensate for precision loss in \eqref{eq:galois_cycle}, we apply
{Newton's} method to $\tilde{y}_{j+1}$ before computing
$\tilde{y}_{j+2}$. Note that at this point, we can check whether
$\tilde{y}_j$ is an approximate root of $p_3(t)$ and whether it lies
on the unit circle.  
Should this not be the case, then either the precision in one of
the previous steps must have been too low; or else one of the conjectures on which our
recipe is based must be false.

Note that it does not matter which of the roots gets the label $y_0$,
as long as the labelling is consistent with the cyclic permutation via
the Galois automorphism $\sigma_m$. This is ensured by using the
Galois polynomial $g_2(t)$ in \eqref{eq:galois_cycle} to order the
roots.  Different choices for $y_0$ lead to fiducial vectors that are
related by Clifford transformations.

\subsection{Minimal polynomial of the square roots}

\noindent
Recall from \eqref{eq:unnormfid} that the components of our
non-normalized fiducial vector are $x_0 = -2-\sqrt{d+1}$ or of the
form $z_j=\sqrt{x_0 y_j}$. When $y_j$ is a root of $p_3(t)$, then
$z_j$ is a root of the polynomial $p_3(t^2/x_0)$. For all our
examples, it turns out that with the particular choice of $x_0$, the
latter polynomial factorises in $H_K[t]$ as
\begin{equation}
  x_0^m p_3(t^2/x_0) = p_4(t)\cdot p_4(-t).\label{eq:square_substitution}
\end{equation}
We pick $p_4(t)$ as the minimal polynomial of the square roots over
the Hilbert class field $H_k$.  The factorisation
\eqref{eq:square_substitution} supports Conjecture~\ref{conj4}, as the
roots $z_j=\sqrt{x_0 y_j}$ of $p_4(t)$ lie in the same field as the
roots $y_j$ of $p_3(t)$.  This also implies that the action of the
Galois automorphism $\sigma_m$ on the roots $z_j$ is cyclic.

\subsection{Fixing the signs}

\noindent
Instead of directly computing the roots of the polynomial $p_4(t)$, we
compute the numerical square roots
$\tilde{z}_j=\sqrt{x_0\tilde{y}_j}$.  This has the advantage that we
automatically get the correct indices corresponding to the cyclic
action of $\sigma_m$.  There is, however, an ambiguity concerning the
sign of the square root. This can be fixed by the observation that
each of the square roots is a factor of only one of the factors in
\eqref{eq:square_substitution}. So we computed
$|p_4(\pm\sqrt{x_0\tilde{y}_j})|$ for both possibilities and chose the
sign which gave the smaller absolute value.  It turned out that a
moderate precision was sufficient for all examples.  As the choice
between $p_4(t)$ and $p_4(-t)$ was arbitrary, we are permitted to flip the sign of
all roots $\tilde{z}_j$. This yields two alternatives:
$\mathbf{z}=(\tilde{z}_0,\dots,\tilde{z}_{m-1})$ and $-\mathbf{z}$.

\subsection{Combinatorial search}
\noindent
According to \eqref{eq:unnormfid}, the first component of the
non-normal\-ized fiducial vector is set to $x_0=-2-\sqrt{d+1}$, while
by \eqref{primetikett}, the other components are computed as
$\hat{\Psi}_{\theta^j \bmod d} = z_{j \bmod m}$. Here $\theta$ is a
primitive element of $\ZZ_d$, i.\,e., it generates the group
$\ZZ_d^\times$ of invertible elements in $\ZZ_d$. When $d$ is prime,
$\ZZ_d^\times$ is cyclic of order $d-1$.  As we do not know the exact
correspondence \eqref{eq:Galois_mat} between elements of the Galois
group and matrices in $GL(2,\ZZ_d)$, we do not know which primitive
element $\theta$ is the correct one.  Hence we try all possibilities,
together with the choice between $\mathbf{z}$ and $-\mathbf{z}$.
Instead of checking all SIC-POVM conditions, it appears to be
sufficient to compute only $|\langle\Psi|X|\Psi\rangle|^2$ and check
whether it is close to $1/(d+1)$. There are $\Phi(d-1)<d/2$ candidates
for $\theta$ (as $d$ is odd) and two choices for the sign, so that we
have to test no more than $d$ candidates for the fiducial vector.

\subsection{Computing an exact fiducial vector}

\noindent
It has to be noted that all exact calculations in the previous steps
of our recipe are carried out in the Hilbert class field{, which in our cases all have sufficiently small degree}. 
We do not have to compute exact roots of the polynomials or
explicit Galois groups of number fields.  Nonetheless, we can obtain
an exact fiducial vector. The first component is $-2-\sqrt{d+1}\in
K$. The other components are roots of the exact polynomial
$p_4(t)$. This allows us to define the number field
\begin{equation}
  \LL=H_K(\gamma), \qquad\text{where $p_4(\gamma)=0$.}\label{eq:exact_field}
\end{equation}
In all examples, the field $\LL$ is the same as the field
${(K^{\pp\jj})}^\tau$ which is defined via $p_2(t)=p_1(t)^\tau$, where
$p_1(t)$ is the minimal polynomial of the Stark units over $K$.  The
mapping
\begin{equation}
  \iota\colon \LL \hookrightarrow \CC, \gamma\mapsto \tilde{z}_0\label{eq:embedding_L}
\end{equation}
extends the chosen embedding of the Hilbert class field $H_K$ to an
embedding of $\LL$.

Similar to our treatment of the Galois polynomial $g_1(t)$ for the
real Stark units, we compute a Galois polynomial $g_4(t)\in H_K[t]$
from the ordered numerical roots of $p_4(t)\in H_K[t]$.  In contrast
to the real Stark units, the numerical values are now complex numbers
so that, for example, the multiplication operation is about three
times slower than that of real numbers.  As the coefficients of
$g_4(t)$ are in the Hilbert class field, we have to find integer
relations among the $2h$ elements of a basis of $H_K/\QQ$, and the
approximation of the coefficient.  This requires that we increase the
precision; but at the same time, there are only $m$ {(instead of
  $mh$)} coefficients.

Setting $z_0=\gamma$, we use the exact Galois polynomial $g_4(t)$ to
compute the exact values
\begin{equation}
  z_{j+1}=g_4(z_j)
\end{equation}
which are components of the fiducial vector.  This yields all exact
roots of $p_4(t)$ without directly factorising the polynomial over the
number field $\LL$ of absolute degree $2mh$ over $\QQ$.  At the same
time, the Galois automorphism $\sigma_m$ is fully specified by
\begin{equation}
  \sigma_m\colon \gamma\mapsto g_4(\gamma).\label{eq:exact_galois}
\end{equation}
As the Galois group of the extension $\LL$ of the Hilbert class field
$H_K$ is cyclic (see also Lemma~\ref{twoprim}) and $H_K$ is totally
real, $\sigma_m^{m/2}$ is the unique automorphism of order $2$ that
corresponds to complex conjugation with respect to the embedding
\eqref{eq:embedding_L}.

When the ray class field has subfields of small degree, one might be
able to obtain a somewhat nicer representation of the field. However,
note that the {degrees} of the fields in Table \ref{tab:prim} have
relatively large prime power factors. When working with the field $\LL$
defined in \eqref{eq:exact_field} via $p_4(t)$ and the Galois
automorphism defined in \eqref{eq:exact_galois}, the arithmetic gets
very slow {as} the degree of the field increases.  That is the reason
why we have not yet computed exact solutions for dimensions $d=1447$,
$2707$, $4099$ {or} $19603$. For dimension $5779$, computing the exact
fiducial vector based on \eqref{eq:exact_galois} takes about two hours.

\subsection{Verifying the solution}
\noindent
Rewriting the SIC-POVM conditions \eqref{eq:sicdefeq} in the form
\eqref{eq:SIC_conditions2} avoids complex $d$-th roots of unity which
would require an additional field extension of degree $\varphi(d)=d-1$
when $d$ is prime. Furthermore, we do not have to normalize the
fiducial vector, which in general requires an additional field
extension of degree two for the square root of the squared norm
$|\langle\hat\Psi|\hat\Psi\rangle|^2$ of the un-normalized
vector. Instead, we work with the un-normalized vector when computing
the sum for $G(i,k)$ in \eqref{eq:Gik1} and divide by
$|\langle\hat\Psi|\hat\Psi\rangle|^2$.  When checking \eqref{Gik} to be zero,
we do not need any normalization.

Nonetheless, full verification of the exact fiducial vector is only
feasible for moderate dimensions, as the number of arithmetic
operations scales with the cube of the dimension.  As already mentioned, the
complexity of the arithmetic in the ray class field $K^{\pp\jj}$
depends on the degree and the representation of the number field. When
the ray class group decomposes into small cyclic factors of prime
power order, one can find somewhat nicer representations of the field
as a compositum of extensions of the Hilbert class field.

For all dimensions for which we have computed an exact solution, we
have computed at least a few of the values $G(i,k)$ in \eqref{eq:Gik1}
or \eqref{Gik} using exact arithmetic.  Timings for {a} single {such} evaluation
are listed in Table \ref{tab:timings_verification}.  When full exact
verification was not feasible, we carried out a numerical check
to the precision given in Table \ref{tab:timings_verification}.

\begin{table}[hbt]
  \caption{Run-time for the verification of the solutions. We list the
    degree of the ray class field $K^{\pp\jj}$ over $K$ as a product
    of prime powers corresponding to the structure of the ray class
    group. The additional factor $2$ for $d=103,487,1447$ is due {to
      the class number $h$ being~$2$.} For $d=2503$, the class number
    is $h=4$. The column $G(i,k)$ provides information on the time
    required to compute one value $G(i,k)$ using exact arithmetic.}
  \medskip
  
  \def\arraystretch{1.1}\tabcolsep2\tabcolsep
  \begin{tabular}{|r|r|r@{ }l|r@{ }l|c|}
    \hline
    \multicolumn{1}{|c|}{$d$} &
    \multicolumn{1}{c|}{$\deg(K^{\pp\jj}/K)$}&
    \multicolumn{2}{c|}{precision} &
    \multicolumn{2}{c|}{CPU time} &
    $G(i,k)$\\
    \hline
     &&&&&&\\[-3ex]
    \hline
     $103$ & $2\times 2\times 17$  &\multicolumn{2}{c|}{exact}  & $440$ & s      & $1.3$ s\\
     $199$ & $2\times 11$          &\multicolumn{2}{c|}{exact}  & $310$ & s      & $0.3$ s\\
     $487$ & $2\times 2\times 3^4$ &\multicolumn{2}{c|}{exact}  & $31$ & days    & $315$ s\\
     $787$ & $2\times 131$           & $10000$ & digits &  $3$ & hours & $65$ min \\
    $1447$ & $2\times 2\times 241$   & $10000$ & digits & $17.0$ & hours & \\
    $2503$ & $4\times2\times3\times139$   &  $10000$ & digits & $87.8$ & hours & \\
    $2707$ & $2\times 11\times 41$   &  $2000$ & digits & $11.2$ & hours & \\
    $4099$ & $2\times 683$           &  $2000$ & digits & $36.5$ & hours & \\
    $5779$ & $2\times 107$           &  $2000$ & digits & $100$  & hours & $88$  min \\
   $19603$ & $2\times 3^2\times 11^2$ & $1000$ & digits & $1367$ & days & \\
    \hline
  \end{tabular}
  \label{tab:timings_verification}
\end{table}

\section{Two detailed examples}\label{sec:IB3}
\noindent 
We will take a closer look at dimensions $d = 7$ and $d =199$. For the
first example $\ell = 1$, for the second example $\ell = 3$, and they
have the same quadratic base field $K = \QQ(\sqrt{2})$ with class
number $h = 1$.  Hence the degree of the small ray class field $K^\mm$
over $K$ is
\begin{equation}
  m = \frac{h(d-1)}{3\ell} =
  \begin{cases}
     2 & \text{for $d = 7$,} \\
    22 & \text{for $d = 199$.}
  \end{cases}
\end{equation}
They are small enough that we can present the calculations in full
detail for~$d=7$ and considerable detail for~$d=199$. No serious
computational difficulties arise. On the other hand all the
complications in principle show up for $d = 199$; except that the
class number $h = 1$.

We begin with $d = 7$. Our starting point is a set of numerical Stark
units, ordered by the ray class group as described in Section
\ref{sec:Stark_units}.  The numerical precision must be made high enough so
that we can determine their exact minimal polynomial, which we call
$p_1(t)$. Its coefficients are integers in the quadratic base field
$K$, and it has degree $m$. For $d = 7$
\begin{equation}
  p_1(t) = t^2 -(1+\sqrt{2})t + 1 .
\end{equation}
We obtain a polynomial $p_2(t)$ whose roots are the Stark phase units
$e^{i\vartheta_j}$ by changing $\sqrt{2}$ to $-\sqrt{2}$. For $d = 7$
\begin{equation}
  p_2(t) = t^2 -(1-\sqrt{2})t + 1 .
\end{equation}
In this example (and the next) the Hilbert class field is identical to
the quadratic field, so step (4) of the recipe is trivial, $p_3(t) =
p_2(t).$ For the fiducial vector we need to take the square roots of
the Stark phase units, but doing so directly would take us out of the
ray class field. To remedy this we introduce the algebraic integer
$x_0 = - 2 - 2\sqrt{d+1}$, and we obtain the minimal polynomial of the
scaled Stark phase units.
For $d = 7$ it reads
\begin{equation}
  x_0^{2}p_3( t/x_0) = t^2 - 2t + 12 + 8\sqrt{2} .
\end{equation}
The roots of this polynomial are of the form $x_0e^{i\vartheta_j}$. We
need the square roots of these numbers. For this purpose we make the
substitution $t \rightarrow t^2$. If the square roots of the roots of
$p_3(t)$ do lie in the small ray class field as predicted by
Conjecture~\ref{conj4}, the resulting polynomial
will factorise over the base field, and indeed it does. For $d = 7$
\begin{equation}
  x_0^2p_3(t^2/x_0) = \Bigl( t^2 + (2+\sqrt{2})t + 2+ 2\sqrt{2}\Bigr)
  \Bigl(t^2 - (2+\sqrt{2})t + 2 + 2\sqrt{2} \Bigr) . \label{eq67}
\end{equation}
To proceed we just pick one of the factors, say the first, and call it
$p_4(t)$ (the other factor is $p_4(-t)$). Our sequence of polynomials
has now reached its end. For $d = 7$
\begin{equation}
  p_4(t) = t^2 + (2+\sqrt{2})t + 2+ 2\sqrt{2} .
\end{equation}
The roots of $p_4(t)$ are of the form $z_j =
\sqrt{x_0e^{i\vartheta_j}}$. The only sign ambiguity which is left is
the overall sign ambiguity that arises when we choose a factor in
eq.~\eqref{eq67}.

Because $p_4(t)$ has only two roots we can write down the SIC fiducial
vector directly. According to eq.~\eqref{primetikett} its
components are distributed using a generator $\theta$ of
$\ZZ_d^\times$, the multiplicative group of integers modulo $d$. For
$d = 7$ there are two possible choices, $\theta = 3$ or $\theta=5$,
but due to the Zauner symmetry they lead to the same vector. We pick
one of the roots of $p_4(t)$ and call it $z_0$. The two possible
choices lead to Clifford-equivalent SICs. Following
Section~\ref{sec:IB1} we define
\begin{alignat}{7}
\hat{\Psi}_{3^0 \bmod 7} = \hat{\Psi}_{3^2 \bmod 7} = \hat{\Psi}_{3^4 \bmod 7} &{}= z_0 ,\\
\hat{\Psi}_{3^1 \bmod 7} = \hat{\Psi}_{3^3 \bmod 7} = \hat{\Psi}_{3^5 \bmod 7} &{}= z_1,
\end{alignat}
where Zauner symmetry was used. The fiducial vector is then
\begin{equation}
  \Psi = N(\pm x_0, z_0, z_0, z_1, z_0, z_1, z_1)^{\rm T} . \label{d7}
\end{equation}
The sign ambiguity is the result of our choice of a factor of
$p_3(t)$. (Instead of choosing the overall sign of the $z_i$, we
choose the sign of the first component $\pm x_0$.) We resolve it by
calculating $\langle \Psi|X|\Psi\rangle$ numerically, and find that
the choice that works in this case is the plus sign. (We have not yet
found a way around this final somewhat inelegant step.) Because $z_0$
is the complex conjugate of $z_1$ it is straightforward to check that
we now have a SIC.

The case $d = 7$ (and $d = 19$) is exceptional because the Galois
group is cyclic of order two, so that no ordering problem
arises. Although we do not need them we can easily calculate the
Galois polynomials. In particular, by solving the quadratic equation
$p_4(t) = 0$ and denoting the roots by $z_0$ and $z_1$ we see that the
polynomial
\begin{equation}
  g_4(t) = 2 + \sqrt{2} - t .
\end{equation}
has the properties that $g_4(z_0) = z_1$ and $g_4(z_1) = z_0$, so this
is the desired Galois polynomial. Incidentally the SIC fiducial
vector \eqref{d7} itself was first derived by hand \cite{Marcus}, by
solving the defining eqs.~\eqref{eq:sicdefeq}.

We turn to $d = 199$. The minimal polynomial of the Stark units,
$p_1(t)$, now has degree $22$, and it can be calculated using Magma in
less than a minute. It is straightforward to go from there to the
polynomial $p_4(t)$ that has the components of the SIC fiducial as
roots. It is
\begin{small}
\begin{alignat}{9}
 p_4&(t)=
      t^{22} + (16-18a)t^{21} + (146 - 54 a) t^{20} - (5848 - 3884 a) t^{19}\nonumber\\
    & - (69328 - 49868 a) t^{18} + (450480 - 320688 a) t^{17} + (9271856 - 6553664 a) t^{16}\nonumber\\
    & - (6276944 -  4448392 a) t^{15} - (511130784 - 361362336 a) t^{14}\nonumber\\
    & - (501114624 - 354768560 a) t^{13} + (12982950368 - 9182919584 a) t^{12}\nonumber\\
    & + (14458284800 - 10213792864 a) t^{11} - (157692490944 - 111463664512 a) t^{10}\nonumber\\
    & - (73845898496 - 52328201280 a) t^9 + (914730317568 - 647071560192 a) t^8\nonumber\\
    & - (136112007424 - 96919409792 a) t^7 - (2444614507008 - 1731374423552 a) t^6\nonumber\\
    & + (1421567308800 - 1043246226432 a) t^5 + (2891624580096 - 1803172507136 a) t^4\nonumber\\
    & - (3497710163968 - 1315439504384 a) t^3 + (1952962604032 +  2282964641792 a) t^2\nonumber\\
    & - (4235310882816 + 5036656351232 a) t +9271967234048 + 7154311792640 a , 
\end{alignat}
\end{small}%
where $a = \sqrt{2}$.  To build the fiducial {vector} we need to order its
roots into a cycle of length $22$. Testing all possibilities is out of
the question, so at this point we need the Galois polynomial
$g_2(t)$. In the recipe that we gave in Section~\ref{sec:Markus} steps
(2), (5), and (9), are concerned with this question. However, for
$d=199$ we can try a shortcut and directly compute the action of the
Galois group for the number field defined by $p_4(t)$ using standard
Magma routines. For the higher dimensions this would take a prohibitive 
amount of time, but in the present case we obtain the answer in a few
minutes. Since the Galois group is cyclic of order $22$ it has $10$
generators of order $22$. We simply pick one of them. Magma then gives
us a polynomial $g_4(t)$ having the properties that
\begin{equation}
  z_{j+1} = g_4(z_j), \quad z_{j+11} = g_4^{[11]}(z_j)=1/z_j, \quad z_{j+22} = g_4^{[22]}(z_j)=z_j, \label{199galois}
\end{equation}
where $g_4^{[k]}(t)=g_4^{[k-1]}(g_4(t))$ denotes the $k$-fold
composition of the polynomial $g_4(t)$ with itself.  The explicit
expression for the polynomial is somewhat unwieldy, and we have
relegated it to Appendix~\ref{sec:AppC}.

We now have to relate this ordering of the roots to the ordering of
the components provided by the Clifford group. The group
$\ZZ_d^\times$ has altogether $\varphi(d-1)=60$ generators $\theta$.
For any choice of generator of the Galois group there will be $6$
choices of $\theta$ that result in a SIC fiducial vector with
$\hat{\Psi}_{\theta^j} = z_j$. There is also a sign ambiguity in
$\hat{\Psi}_0$.  We resolve these ambiguities by a numerical
calculation of a single overlap $\langle \Psi|X|\Psi\rangle$.  When
using $p_4(t)$ and the Galois polynomial $g_4(t)$ one finds that we
have to choose a negative sign for the zeroth component, while the
appropriate choices of $\theta$ are $\theta = 41, 75, 134, 167, 189,
190$. Because of the symmetry these choices result in identical
vectors.

We now have a candidate exact SIC fiducial vector for $d = 199$. In
this case the exact calculation that verifies that we really have a
SIC is quick, as can be seen in Table \ref{tab:timings_verification}.

The computational complexity of some of the steps described in
Section~\ref{sec:Markus} grows with the dimension, but the logic
remains simple.  All our solutions can be described simply in words:
The components of the fiducial vector are square roots of scaled
Stark phase units ordered by the Galois group of the small ray class
field over the Hilbert class field.

\section{Remarks, observations, and extensions}\label{sec:IB4}

\noindent It is time to compare our SIC recipe with other methods that
have been used to construct them.  The first extensive collection of
exact solutions was obtained by solving the defining equations
\eqref{eq:sicdefeq} using Gr\"obner bases \cite{Scott}. After that the
majority of new solutions have been obtained by applying
integer relation algorithms to high precision numerical solutions,
making use of (and supporting) the conjectures about the relevant
number fields \cite{ACFW}. By now exact solutions in more than one hundred
different dimensions have been obtained using this method and
variations\cite{MG, MG2}, including many dimensions of the form
$d=n^2+3$, namely every such dimension with $n \leq 18$ together with
$d= 403$, $844$, and $1299$. But this method encounters two
bottlenecks. One is that it becomes necessary to factorise polynomials of
high degree over number fields of high degree. The other is the very
time consuming search for a numerical solution to start with
\cite{Andrew}. Prior to the work reported here, $d = 1299$ was the
highest dimension in which a SIC was known in exact form
\cite{MG}. The highest dimension for which a numerical search has been
successful\cite{AS} is $d = 2208$.  For $d = 5779$ there has been an
attempt to find a numerical solution with the standard routines. After
$55065$ trials and $17.69$ CPU years the search came up empty-handed
\cite{MG2}. Both bottlenecks are avoided by our recipe, which is why
we have been able to reach higher dimensions than ever before.

We add two observations that lead to open questions.  The first
concerns a theorem that we were unable to prove. Consider the overlap
phases
\begin{equation}
  e^{i\nu_j} = \sqrt{d+1}\langle \Psi|X^j|\Psi\rangle . \label{babies}
\end{equation}
They are phase factors because $|\Psi\rangle$ is a SIC fiducial
vector, and they belong to the small ray class field because the
operator $X$ contains only real integers (in fact, its entries are $1$
or $0$). Because of the symmetry of the fiducial vector the number of
distinct phase factors equals $(d-1)/3\ell$. Hence one may suspect
that these overlap phases are identical to the Stark phase units.  Let
$a_j$ be the $d-1$ non-trivial components of the fiducial vector, as
in eq.~\eqref{Psi1}. 
In our examples we have observed that those overlaps are indeed Stark
phase units, and moreover
\begin{equation}
\sqrt{d+1}\langle \Psi|X^{-2j}|\Psi \rangle = -\frac{a_j^2}{|a_j|^2} .
\end{equation}
A proof that this formula must hold would strengthen the motivation
for our recipe, and might offer insight into the structure underlying
the Stark units.

The second observation concerns the real fiducial vectors that can be
obtained by a discrete Fourier transformation from the ones we have
constructed. By inspection one can see that real fiducial vectors are
largely built from real units \cite{Anti1}, but not in any obvious way
from Stark units. Still it may in the end prove advantageous to work
with real rather than complex fiducials, especially since the ordering
of the real Stark units is directly computationally accessible. 

Finally, let us comment on the cases when $d$ is not a prime
number. When $d$ is odd and the multiplicities of its prime factors
equal one, the Clifford group splits into a direct product, and each
prime factor in $d$ splits into prime ideals. We now have to compute
Stark units for an entire lattice of subfields. For each subfield we
have to relate the action of the Galois group to the permutation
action on the fiducial vector provided by the Clifford group, and the
choices have to be correlated.  There is also an issue of principle
that arises when the dimension is divisible by $3$, which is that we
encounter a Zauner symmetry of a different type (of type $F_a$ in the
terminology of Scott and Grassl \cite{Scott}). We feel that we have
not yet found the best way to deal with these problems, and we have
decided to postpone a full presentation of these results. When the
dimension $d$ is even it will contain a factor $4$, and this factor
has to be given special treatment. Having done this we find that there
is---after a change of basis---again an almost flat fiducial vector
whose components can be obtained from Stark units{;} but this time there
are some significant differences {from} the case of odd $d$.  As it
happens the degrees of the small ray class fields that occur are in
some ways more manageable than those that occur for odd $d$. Again we
postpone a full presentation of these results. We can, however, report
already now that we have used the connection to the Stark units to
write down exact or numerical SIC fiducial vectors in dimensions of
the form $d=n^2+3$ for 
$d=4$,
$7$,
$12$,
$19$,
$28$,
$39$,
$52$,
$67$,
$84$,
$103$,
$124$,
$147$,
$172$,
$199$,
$259$,
$292$,
$327$,
$403$,
$487$,
$579$,
$679$,
$628$,
$787$,
$844$,
$964$,
$1027$,
$1159$,
$1228$,
$1299$,
$1447$,
$1603$,
$1684$,
$1852$,
$1939$,
$2119$,
$2307$,
$2404$,
$2503$,
$2707$,
$3028$,
$3603$,
$4099$,
$4492$,
$4627$,
$5779$,
$6727$,
$7399$,
$19603$, and
$39604$.

\begin{table}[hbt]
\caption{Prime dimensions in the sequence $d = n^2+3$, including all
  primes among the first $100$ entries.  The {columns detail the} class
  number $h$, {the} position $\ell$ in the sequence of dimensions
  \eqref{deez}, the factorised absolute degree of the small ray class
  field, the height (absolute value of the largest coefficient) of the
  minimal polynomial of the Stark units (represented by the $\log_{10}$, when
  known or estimated), and a note telling if the full verification of
  the SIC property is numerical or exact. The dimension is in boldface
  if the relevant SIC has been constructed using our recipe.}
  \medskip

  \def\arraystretch{1.3}\tabcolsep2\tabcolsep
  \begin{tabular}{|r|r|r|r|c|r|c|}
    \hline
    \multicolumn{1}{|c|}{$n$} &
    \multicolumn{1}{c|}{$d$} &
    \multicolumn{1}{c|}{$h$} &
    \multicolumn{1}{c|}{$\ell$} &
    degree &
    \multicolumn{1}{c|}{log height} &
     check\\
     \hline
     &&&&&&\\[-3.7ex]
     \hline
     {\bf 2} & {\bf 7} & 1 & 1 & $2^2$ & 1 & exact \\
     {\bf 4} & {\bf 19} & 1 & 3 & $2^2$ & 1 & exact \\
     {\bf 8} & {\bf 67} & 1 & 1 & $2^2\times 11$ & 10 & exact \\
     {\bf 10} & {\bf 103} & 2 & 1 & $2^3\times 17$ & 44 & exact \\
     {\bf 14} & {\bf 199} & 1 & 3 & $2^2\times 11$ & 11 & exact \\
     {\bf 22} & {\bf 487} & 2 & 1 & $2^2\times 3^4$ & 424 & exact \\
     {\bf 28} & {\bf 787} & 1 & 1 & $2^2\times 131$ & 299 & $10000$ digits \\
     {\bf 38} & {\bf 1447} & 2 & 1 & $2^3\times 241$ & 2158 & $10000$ digits \\
     {\bf 50} & {\bf 2503} & 4 & 1 & $2^4\times 3 \times 139$ & 6464 & $10000$ digits \\
     {\bf 52} & {\bf 2707} & 1 & 1 & $2^2\times 11 \times 41$ & 1861 & $2000$ digits \\
     62 & 3847 & 4 & 1 & $2^4\times 641$ & 11133 &  \\
     {\bf 64} & {\bf 4099} & 1 & 1 & $2^2\times 683$ & $974$ & $2000$ digits \\
     70 & 4903 & 10 & 1 & $2^3\times 5\times 19\times 43$ & 25224 & \\
     74 & 5479 & 4 & 1 & $2^4\times 11 \times 83$ & 19618 & \\
     {\bf 76} & {\bf 5779} & 1 & 9 & $2^2\times 107$ & $127$ & $2000$ digits \\
     92 & 8467 & 2 & 1 & $2^3\times 17 \times 83$ & 12133 &  \\
     94 & 8839 & 8 & 1 & $2^5\times 3\times 491$ & 48203 &  \\
     \hline
     {\bf 140} & {\bf 19603} & 1 & 3 & $2^2\times 3^2\times 11^2$ & $1754$ & $1000$ digits \\
     \hline
  \end{tabular}
  \label{tab:prim}
\end{table}

\section{Conclusions}\label{sec:summary}

\noindent We have proposed a recipe for constructing SICs in prime dimensions
of the form $d = n^2+3$. The key ingredients of the construction are the (conjectural) Stark
units in precisely specified abelian extensions of real quadratic number fields.

The remarkable interplay between the geometry and the
number theory began \cite{AFMY} with the fact that a search for SICs in
complex \emph{dimension} $d$ led directly to a ray class field of
\emph{modulus} $d$.  It is tempting to make the analogy with the
cyclotomic case (a polygon with $d$ sides leads to a ray class field
over $\QQ$ of conductor $d$); but 
as yet we know of no abelian variety whose $d$-torsion points would lead to this structure.

In the prime dimensions of concern in this paper we are able to find `minimal' fiducial vectors 
having components that lie inside a much smaller ray class field, whose conductor has finite part
$\pp$ {being} just one of the factors $\sqrt{p+1}\pm 1$ of $p$ inside the ring
of integers $\ZZ_K$ of the base field $K$.  Among other things, this
facilitates calculations as the degree is significantly smaller than 
that of the full `overlap field' explored
in Refs. \citenum{AFMY,Kopp,Kopp2}.

But possibly the most remarkable feature of this construction is that
we produce a fiducial vector whose non-trivial entries are
\emph{precisely} the set of Galois conjugates of the units described
in the most ambitious form of Stark's original
conjectures\cite{stark4}.  Namely, those entries are the \emph{square
roots} of the units derived directly from Stark's original predicted
$L$-function values, about which Tate~\cite{tatetok} --- combining
this with work of Brumer --- later broadened the conjecture, to
predict that these square roots would still live in an \emph{abelian}
extension of $K$.  It is an additional strange quirk that in every
case we have considered, this top field is identical to that which
would be created by taking the original ray class field of conductor
$\pp$ and adjoining the geometric scaling factor (see
Conjecture~\ref{conj4}).

Once the Stark units have been computed, the outline of the recipe, as
given in Section \ref{sec:intro}, is very simple. A precise version of
the recipe, including all the calculational steps, was given in
Section \ref{sec:Markus}. As can be seen there it leads to very
complex calculations. It should not surprise anyone that working out
explicit formulas in dimension $4099$ (say) leads to complex
calculations, but it is easy to see why it is nevertheless important
to do so.

One reason is that we have at present no suggestions for how to turn
our recipe into a general existence proof for the infinite sequence of
dimensions we consider.  If indeed it requires a version of the
complex Stark conjectures then at present it would seem quite
inaccessible.  However since it is possible to formulate a version of
Zauner's SIC existence conjecture over $p$-adic fields and examine it
prime-by-prime, it seems appropriate to mention the huge recent
progress on $p$-adic versions of the (Brumer-)Stark Conjectures.  It
would take us too far afield to try to incorporate a proper account
into this paper, but broadly speaking there are two strands to this.
See Refs. \citenum{daskak1,daskak2} and the more recent survey article
\citenum{daskakICM} for the celebrated recent proof of the conjectures
and the context into which it fits.  On the other hand, Darmon, Pozzi
and Vonk have made considerable inroads into the same questions by
very different techniques.  The latest in this line of papers may be
found at Ref. \citenum{DPV} which is also placed into the overall
context in the paper \citenum{daskakICM}.  These may conceivably
provide a `tunnel' into our questions via an argument using some sort
of Hasse local-global principle, whereby one may deduce the existence
of a solution to the SIC problem via that of~$p$-adic solutions at
each~$p$.

For now, SIC existence is proved by explicit construction, dimension
by dimension.  It therefore becomes important to guard against
low-dimensional accidents.  The problem is that we do not know what
``low'' means in this context.  We feel, however, that once we get to
the four-digit dimensions we are probably on the safe side. See Table
\ref{tab:prim} for a summary of the results we have achieved in prime
dimensions.  Additional data can be found online \cite{online}.

We note that there is only one case ($d = 487$) in which we have
carried out a complete exact verification of the SIC condition for a
SIC that was not previously known.  In the remaining cases the exact
verification was at best partial, but complete numerical checks with
high precision have been performed in all cases (see Table
\ref{tab:timings_verification}).  However, it seems to us more
important to say that we have found $13$ prime dimensions (and $49$
dimensions altogether, as listed at the end of Section \ref{sec:IB4})
of the form $d = n^2+3$ for which a SIC fiducial vector can be built
from Stark units, and none where the recipe for doing this fails.

\section*{Acknowledgements}
Thanks to John Coates and Henri Darmon for encouraging insights during the course of this work. 
We also thank the anonymous referee for many helpful remarks. 
Access to computing resources of the divisions Marquardt and Russell
of the Max Planck Institute for the Science of Light, as well as the
HPC clusters Raven and Cobra of the Max Planck Computing and Data
Facility, {are} {gratefully} acknowledged.  The `International Centre for Theory of
Quantum Technologies’ project (contract no.~MAB/2018/5) is carried out
within the International Research Agendas Programme of the Foundation
for Polish Science co-financed by the European Union from the funds of
the Smart Growth Operational Programme, axis IV: Increasing the
research potential (Measure 4.3). We also thank the Mathematics
Department at Stockholm University for access to their computer. 
GM thanks the QOLS Group at Imperial College London for their ongoing hospitality.

\clearpage

\appendix

\section{\protect\hyperlink{Appendix \protect\ref{sec:AppB}}{}}\label{sec:AppB}

Here we give a brief summary of the Weyl--Heisenberg and Clifford groups. 
For a complete account see, e.\,g., Ref.~\citenum{Marcus}. Given a positive 
integer $d$ there is a Weyl--Heisenberg group generated by $X$, $Z$, and $\omega$, 
subject to 
\begin{equation}
  \omega^d = X^d = Z^d = {\bf 1} , \quad ZX = \omega XZ , \quad
  \omega X = X \omega , \quad \omega Z = Z \omega .
\end{equation}
This is an example of a Heisenberg group in the sense of
Chapter~I of Ref.~\citenum{tatamum}.  It has an essentially unique
unitary representation in dimension $d$, such that $\omega$ is
represented by $e^\frac{2\pi i}{d}$ times the unit matrix, and such
that $Z$ is a diagonal $d\times d$ matrix \cite{Weyl}. The basis
vectors are labelled by integers modulo $d$ and are ordered according
to
\begin{equation}
  Z|r\rangle = \omega^r|r\rangle , \qquad X|r\rangle = |r+1\rangle.
\end{equation} 
Throughout this paper we assume that this basis has been
fixed. Note that this means that the labeling of the $d$ components of
our vectors goes from $0$ to $d-1$.

The Clifford group is defined as the group of automorphisms of the
Weyl--Heisenberg group within the unitary group. To discuss this it is
convenient to collect pairs of integers modulo $d$ (or $2d$ if $d$ is
even) into a `vector' ${\bf p} = (a,b)$, and to define the
displacement operators
\begin{equation}
  D_{\bf p} = D_{a,b} = \left( - e^{\frac{\pi i}{d}}\right)^{ab} X^aZ^b.
\end{equation}
When $d$ is odd the phase factor is a $d$-th root of unity.
It can now be shown that the Clifford group contains the
representation of a symplectic group.  More precisely, for every
symplectic two by two matrix $M$ with entries that are integers modulo
$d$ (or $2d$ if $d$ is even) there is a unitary representative $U_M$
determined up to an overall phase such that
\begin{equation}
  M = \left( \begin{array}{cc} \alpha & \beta \\ \gamma & \delta \end{array} \right)
 \quad \Longrightarrow \quad 
 U_MD_{\bf p}U_M^{-1} = D_{M{\bf p}}.
\end{equation}
The matrix $M$ is symplectic if $\det{M} = \alpha \delta -
\beta \gamma = 1$ modulo $d$ (or $2d$ if $d$ is even).  Importantly,
the representation of this symplectic group is determined up to
overall phase factors by the representation of the Weyl--Heisenberg
group. Moreover the matrix elements of the unitaries $U_M$ are $d$-th
(or $2d$th) roots of unity up to an overall factor which also belongs
to the cyclotomic field.  The matrices $U_M$ are referred to as
\emph{symplectic unitaries} in the literature.

An example of a symplectic unitary that is given in the text is the
Fourier matrix $U_F$ in eq.~\eqref{eq:Fourier}. The square of the Fourier
matrix is a permutation matrix.  For the corresponding symplectic
matrices it holds that
\begin{equation}
  F = \left( \begin{array}{cc} 0 & - 1 \\ 1 & 0 \end{array} \right) , \qquad
 P = F^2 = \left( \begin{array}{cc} - 1 & 0 \\ 0 & -1 \end{array} \right) , \qquad
 \left( U_P\right)_{r,s} = \delta_{r+s,0}. \label{parity}
\end{equation}
The $\delta$ is a Kronecker delta modulo $d$. The operator
$U_P$ whose matrix elements are defined by it is known as the parity
operator. Importantly, every diagonal symplectic matrix is represented
by a permutation matrix that leaves the first component $\Psi_0$ of
any vector invariant. An example is the order $3\ell$ symplectic
matrix $F_S$ given in eq.~\eqref{eq:Fourier}. The corresponding
symplectic unitary $U_S$ is a permutation matrix leaving the fiducial
vector $|\Psi\rangle$ invariant.

A strengthened version \cite{Marcus} of Zauner's conjecture
\cite{Zauner} states that for every SIC that is an orbit under the
Weyl--Heisenberg group there exists a transformation in the Clifford
group so that one of the resulting vectors is invariant under a
symplectic unitary of order three, having the additional property that
the trace of its symplectic matrix equals $-1$ modulo $d$ (or $2d$ if
$d$ is even).  The group of symplectic matrices whose unitary
representatives stabilize a given SIC fiducial vector is denoted
$S$. Unitary symmetries of order $3\ell$ occur in dimensions that sit
higher up in the dimension sequences described in eq.~\eqref{deez} of
Section~\ref{sec:towers}.

The extended Clifford group is the automorphism group of the Weyl--Heisenberg group 
within the group of all unitary and anti-unitary transformations in the Hilbert space. 
Recall that, given a fixed basis, an anti-unitary transformation is represented by 
complex conjugation followed by a unitary transformation. It turns out that 
complex conjugation is the representative of the anti-symplectic matrix 
\begin{equation}
  J = \left( \begin{array}{cc} 1 & 0 \\ 0 & -1 \end{array} \right) . 
  \label{J}
\end{equation}
In its defining $2\times 2$ representation the extended symplectic group 
is obtained by adding $J$ as a generator to the symplectic group. Such a matrix 
is represented by an anti-unitary transformation if its determinant equals $-1$ modulo 
$d$ (or $2d$ if $d$ is even).

\clearpage

\section{\protect\hyperlink{Appendix \protect\ref{sec:AppC}}{}}\label{sec:AppC}

The Galois polynomial that is needed in eq.~\eqref{199galois} to specify the
exact solution in $d = 199$ is, explicitly,
\begin{small}
\begin{alignat}{6}
  g_4(t) =
\frac{1}{2^{10}\cdot 7^{20} \cdot 97 \cdot 137 \cdot 353 \cdot 11777}  \nonumber\\
\times\Bigl( 2^{10}7^{20}(-17078833449878756 + 12080287801165569 a)
 \nonumber\\
   + 2^{10}7^{20}(40287696684934067 - 28489493217381300 a) t  \nonumber\\
   + 2^{9}7^{19}(234689455794798196 - 165942513813988969 a) t^2  \nonumber\\
   - 2^{9}7^{19} (1267244797948364094 - 896075533706946133 a) t^3  \nonumber\\
   + 2^97^{17} (9850430584095374136 - 6965297266085092019 a) t^4  \nonumber\\
   + 2^87^{16} (707445743433268037312 - 500239679098868649349 a) t^5  \nonumber\\
   - 2^77^{15} (2833346680504939765576 - 2003478630481275710517 a) t^6  \nonumber\\
   - 2^{8}7^{14} (13062826273186848404584 - 9236813105616767488891 a) t^7   \nonumber\\
   + 2^77^{13} (50446588512588102272829 - 35671125113271678931917 a) t^8   \nonumber\\
   + 2^67^{12} (469995705600596099309418 - 332337149050655917780371 a) t^9  \nonumber\\
   - 2^67^{11}(639261563754521585705288 - 452026182486419826134005 a) t^{10}   \nonumber\\
   - 2^57^{10} (4169521749102877568698537 - 2948297111467613829154913 a) t^{11}  \nonumber\\
   + 2^47^{9} (4761186552625810725476246 - 3366667317059104479408941 a) t^{12}  \nonumber\\
   + 2^47^{9}(2462129766016039342770763 - 1740988649521184466334136 a) t^{13}  \nonumber\\
   - 2^47^{7} (674236002730880041782142 - 476756827271704087500383 a) t^{14}  \nonumber\\
   - 2^37^{6} (31996193193391918129178252 - 22624725168026626035121663 a) t^{15}  \nonumber\\
   - 2^37^{5}(8855058917437803643379503 - 6261472753616669980085918 a) t ^{16}  \nonumber\\
   + 2^47^4 (6119378614605756530360451 - 4327055148053001825770557 a) t^{17}  \nonumber\\ 
   + 2^27^{3} (12639797088162032274244852 - 8937678964369046837867715 a) t^{18}  \nonumber\\
   - 2^27^{2} (2094260756589813076610143 - 1480856690531237031689199 a) t^{19}  \nonumber\\
   - 7(9470731731247123181642606 - 6696848429749782828833133 a) t^{20}  \nonumber\\
   - (1563447700882137315845205 - 1105520045182155035071680 a) t^{21} \Bigr) ,\label{eq:galois199}
\end{alignat}
\end{small}
where $a = \sqrt{2}$.

\end{document}